\documentclass[10pt,journal,compsoc]{IEEEtran}

\ifCLASSINFOpdf
\else
 \fi

\usepackage{amsthm}
\usepackage{amsmath}
\usepackage{graphicx}
\usepackage[algo2e]{algorithm2e} 
\usepackage[compatible]{algpseudocode}
\usepackage{amssymb}
\usepackage{amsfonts}
\usepackage{algorithm}
\usepackage{algcompatible}
\usepackage{tabularx}
\usepackage{lipsum}

\usepackage[utf8]{inputenc}
\newtheorem{theorem}{Theorem}

\usepackage{url}
\usepackage{xcolor}
\usepackage{balance}
\usepackage{subcaption}
\usepackage{epstopdf}
\usepackage{array}
\usepackage{comment}
\usepackage[noadjust]{cite}
\usepackage{amsthm}
\usepackage{tipa}
\usepackage{breqn}

\SetKw{Continue}{continue}
\SetKw{Break}{break}
\SetKwFor{RepTimes}{repeat}{times}{end}

\theoremstyle{definition}
\newtheorem{definition}{Definition}

\begin{document}
\title{BPFISH: Blockchain and Privacy-preserving FL Inspired Smart Healthcare}

\author{Moirangthem~Biken~Singh, and
	Ajay~Pratap,~\IEEEmembership{Member,~IEEE}
	\IEEEcompsocitemizethanks{\IEEEcompsocthanksitem M. B. Singh and A. Pratap are with the Department of Computer Science and Engineering, Indian Institute of Technology (Banaras Hindu University) Varanasi 221005 India. E-mail: \{moirangthembsingh.rs.cse21, ajay.cse\}@iitbhu.ac.in.\protect }
}

\IEEEtitleabstractindextext{
\begin{abstract}
This paper proposes Federated Learning (FL) based smart healthcare system where Medical Centers (MCs) train the local model using the data collected from patients and send the model weights to the miners in a blockchain-based robust framework without sharing raw data, keeping privacy preservation into deliberation. We formulate an optimization problem by maximizing the utility and minimizing the loss function considering energy consumption and FL process delay of MCs for learning effective models on distributed healthcare data underlying a blockchain-based framework. We propose a solution in two stages- first, offer a stable matching-based association algorithm to maximize the utility of both miners and MCs and then solve loss minimization using Stochastic Gradient Descent (SGD) algorithm employing FL under Differential Privacy (DP) and blockchain technology. Moreover, we incorporate blockchain technology to provide tempered resistant and decentralized model weight sharing in the proposed FL-based framework. The effectiveness of the proposed model is shown through simulation on real-world healthcare data comparing other state-of-the-art techniques.

\end{abstract}

\begin{IEEEkeywords}
Federated Learning, Blockchain, Stable Matching, Differential Privacy, Smart Healthcare.
\end{IEEEkeywords}
}

\maketitle

\IEEEdisplaynontitleabstractindextext

\IEEEpeerreviewmaketitle
   

\section{Introduction}\label{sec:introduction}

Smart healthcare system is likely to play a crucial role in our society. It can provide remote medical services, helps medical diagnosis and protect patients against dangerous infectious disease such as COVID-19 during personal visits to hospitals.
Smart healthcare solutions are extremely helpful during the recent COVID-19 pandemic \cite{singh2020internet}, \cite{jnr2020use}.
However, there are still challenges in smart healthcare such as data unavailability due to privacy concerns of patients \cite{10.1145/3298981}.

Federated Learning (FL) \cite{mcmahan2017communication}, an emerging framework is used to address the unavailability and privacy risks of sensitive healthcare data \cite{KimberlyPowellBlog, 8931716, li2020federated}. Since training data is not leaving the Medical Centres (MCs), this framework ensures privacy protection for all involved centres by sharing models instead of sharing sensitive raw healthcare data. However, there are still challenges in FL based healthcare architecture such as privacy leakage and single point of failure \cite{10.1145/3298981, 8733825}. 
Firstly, healthcare data is highly privacy-sensitive, thus the leakage of this sensitive information could destroy the reputation and finances of the patient. 
Secondly, the existing FL based healthcare architectures are suffering from a single point of failure risks due to the aggregation of the model in a central server and the unwillingness of MCs to participate in the FL process due to the lack of incentives.

To overcome the above-mentioned issues, we consider a blockchain and FL based smart healthcare framework (see Fig. \ref{application_model}) in which MCs train a model locally using the data collected from patients and forward the model weights to a miner in a blockchain to build a robust model without sharing raw data. 
However, it is necessary to optimize both the utilities of miners as well as MCs and the FL loss function simultaneously to increase the accuracy and effectiveness of smart healthcare system while keeping privacy risks into consideration. 
Therefore, we propose a joint optimization problem by maximizing the utility and minimizing the loss function together, considering energy consumption and the model training delay at MCs for learning effective models. However, in blockchain-based FL, there is a need for a proper association between miners and MCs in order to increase accuracy and effectiveness of smart healthcare system.
Thus, we offer a stable association algorithm between miners and MCs to maximize the utility of both miners as well as MCs in polynomial time and then solve the FL loss minimization using Stochastic Gradient Descent (SGD) algorithm employing FL under Differential Privacy (DP) and blockchain technology.

In this paper, we focus on developing a blockchain-based privacy-preserving FL framework for collaborative training across multiple MCs under differential privacy to solve privacy leakage problems in Healthcare Domain (HD). 
An adversary can use an inference attack to recreate the training data from the shared model.
Therefore, we incorporate DP \cite{10.1007/11681878_14} to prevent privacy leakage while transmitting model weights during FL process. 
DP is a privacy protection technique that has been widely used in the field of privacy protection in deep learning models \cite{10.1007/11681878_14, Abadi_2016, wang2019collecting}. 
We also use blockchain as a way for decentralized architecture in model weights sharing. Blockchain provides a tempered resistant framework where each device verifies the transactions in the network and the miners perform Proof-of-work (PoW) to add new blocks. PoW is a proof of doing computation by the miner to add a new block to the blockchain. 
Therefore, we formulate a blockchain and FL based privacy-preserving for smart healthcare framework by maximizing the utility and minimizing the loss function considering energy consumption and latency of MCs for learning effective models. 
Furthermore, to improve the utility and accuracy of the FL model, we offer a stable and computationally efficient association algorithm among miners and MCs. 
We also proposed an incentive mechanism to promote the MCs for participation in the FL process.
Specifically, the contributions of this paper are summarized as follows:
\begin{itemize}

    \item Formulate an optimization problem by maximizing the utility and minimizing the loss function considering energy consumption and FL process delay of MCs for collaborative learning on distributed healthcare data. 
    
    \item Stable Miner-MC Association (MMA) algorithm is proposed between miners and MCs to maximize the utility with computational complexity of $O(N^2S)$, where $N$ and $S$ represent the number of MCs and the number of miners, respectively. 
    
    \item Blockchain-based privacy-preserving FL framework that guarantees DP for decentralized collaborative learning from data stored across MCs is proposed, with communication complexity of $O(T|\mathbb{K}|^{2}|w|)$, where $T$, $|\mathbb{K}|$ and $|w|$ represent the number of global iteration, the number of participated MCs in the FL process and the number of model weights, respectively.
    An Incentive Mechanism (IM) is also proposed to encourage miners for verifying and adding the blocks, and MCs for participating as proportional to their data sample size.
    
    \item Through rigorous evaluations on Chest X-ray Images (Pneumonia) dataset we verify the effectiveness of the proposed model on various privacy parameters. Performance  study  demonstrates that proposed  BPFISH framework  surpass state-of-the-art  schemes, achieving 11.18\% better outcome on an average.

\end{itemize}

The rest of the paper is organized as follows: The relevant work is reviewed in Section \ref{RelWork}. The system model and the problem formulation are discussed in Section \ref{system}.
Miner-MC association and FL process are given in Section \ref{matching} and Section \ref{FL}, respectively. Experiment results and analysis are given in Section \ref{result}. Finally, Section \ref{conc} provides the conclusion.

\section{Related Work} \label{RelWork}

In this section, we present the related works for FL with DP, blockchain, association and their combination in HD along with a comparative analysis in Table \ref{comparisontable}.

\subsection{FL with DP in Smart Healthcare}

Recently, NVIDIA introduced Clara FL \cite{KimberlyPowellBlog} for distributed, collaborative AI model training across multiple hospitals to develop a more accurate global model while maintaining patient privacy.
In \cite{tuli2020healthfog} and \cite{chen2020fedhealth}, authors developed a FL model for IoT based smart healthcare system.
Silva et al. \cite{silva2019federated} also proposed an FL framework for securely accessing and analysing medical data stored at several institutions. These works consider central server for model aggregation in FL process but this causes inaccurate global model update if the server is malfunctioned or attacked. 
The adaptation of privacy-preserving techniques to conserve patient privacy with the use of FL has been demonstrated in clinical and epidemiological research \cite{sadilek2021privacy}. 
In \cite{9492000} and \cite{LI2020101765}, FL based privacy-preserving method has been used for smart healthcare systems such as detection of Alzheimer's disease and multi-site fMRI data analysis. 
In \cite{9069945}, K. Wei et al. proposed a framework under DP by adding Gaussian noises to the locally trained weights before sending it to the server. 
However, these works mainly focused on optimizing FL under DP by considering centralized model aggregation that is vulnerable to the single point of failure and did not consider incentives for participating in the FL process.

\begin{table}[t]
     \begin{center}
         \caption{Summary of existing works} \label{comparisontable}
        
         \begin{tabular}{| c | c | c | c | c | c | c |}
         \hline
         \textbf{{Problem Focus}} &
        \textbf{{FL}} &
        \textbf{{DP}} & \begin{tabular}{@{}c@{}}\textbf{Block}\\ \textbf{-chain}\end{tabular} & \begin{tabular}{@{}c@{}}\textbf{Associa}\\ \textbf{-tion}\end{tabular} &
        \textbf{{IM}} &
        \textbf{{HD}} \\
        \hline
        
            \begin{tabular}{@{}c@{}}\textbf{Collaborative}\\ \textbf{AI Model \cite{KimberlyPowellBlog}}\end{tabular} & {\checkmark} & {$\times$} & {$\times$} & {$\times$} & {$\times$} & {\checkmark} \\ 
            \hline
            \begin{tabular}{@{}c@{}}\textbf{Smart Health-}\\\textbf{care \cite{tuli2020healthfog, chen2020fedhealth}}\end{tabular} & {\checkmark} & {$\times$} & {$\times$} & {$\times$} & {$\times$} & {\checkmark} \\ 
            \hline
            \begin{tabular}{@{}c@{}}\textbf{Meta-Analysis}\\\textbf{of Brain}\\\textbf{Data \cite{silva2019federated}}\end{tabular} & {\checkmark} & {$\times$} & {$\times$} & {$\times$} & {$\times$} & {\checkmark} \\ 
            \hline
            \begin{tabular}{@{}c@{}}\textbf{Health}\\ \textbf{Research \cite{sadilek2021privacy}}\end{tabular} & {\checkmark} & {\checkmark} & {$\times$} & {$\times$} & {$\times$} & {\checkmark} \\ 
            \hline
            \begin{tabular}{@{}c@{}}\textbf{Privacy}\\ \textbf{Preserving \cite{9492000}}\end{tabular} & {\checkmark} & {\checkmark} & {$\times$} & {$\times$} & {$\times$} & {\checkmark} \\ 
            \hline
            \begin{tabular}{@{}c@{}}\textbf{Multi-site}\\ \textbf{fMRI}\\ \textbf{Analysis \cite{LI2020101765}} \end{tabular} & {\checkmark} & {\checkmark} & {$\times$} & {$\times$} & {$\times$} & {\checkmark} \\ 
            \hline
            
            \begin{tabular}{@{}c@{}}\textbf{Prevent}\\ \textbf{Information}\\ \textbf{Leakage \cite{9069945}} \end{tabular} & {\checkmark} & {\checkmark} & {$\times$} & {$\times$} & {$\times$} & {$\times$} \\ 
            \hline
            
            \begin{tabular}{@{}c@{}}\textbf{Access}\\ \textbf{Management \cite{novo2018blockchain}} \end{tabular} & {$\times$} & {$\times$} & {\checkmark} & {$\times$} & {$\times$} & {\checkmark} \\ 
            \hline
            \begin{tabular}{@{}c@{}}\textbf{EMRs Access}\\ \textbf{Control \cite{zhang2018blockchain,dagher2018ancile}} \end{tabular} & {$\times$} & {$\times$} & {\checkmark} & {$\times$} & {$\times$} & {\checkmark} \\ 
            \hline
            \begin{tabular}{@{}c@{}}\textbf{EMRs}\\ \textbf{Management \cite{9261429}} \end{tabular} & {$\times$} & {$\times$} & {\checkmark} & {$\times$} & {$\times$} & {\checkmark} \\ 
            \hline
            
            \begin{tabular}{@{}c@{}}\textbf{Resource off-}\\ \textbf{loading \cite{9112693, lim2021towards}} \end{tabular} & {$\times$} & {$\times$} & {\checkmark} & {\checkmark} & {$\times$} & {$\times$} \\ 
            \hline
            \begin{tabular}{@{}c@{}}\textbf{Blockchain}\\ \textbf{based Crowd-}\\ \textbf{sourcing \cite{KADADHA2021103155, 8726129}} \end{tabular} & {$\times$} & {$\times$} & {\checkmark} & {\checkmark} & {$\times$} & {$\times$} \\ 
            \hline
            \begin{tabular}{@{}c@{}}\textbf{Patient-}\\ \textbf{Physician}\\ \textbf{Matching \cite{9176324}} \end{tabular} & {$\times$} & {$\times$} & {\checkmark} & {\checkmark} & {$\times$} & {$\times$} \\ 
            \hline
            
            \begin{tabular}{@{}c@{}}\textbf{Privacy}\\ \textbf{Preserving}\\ \textbf{FL \cite{lu2019blockchain}} \end{tabular} & {\checkmark} & {\checkmark} & {\checkmark} & {$\times$} & {$\times$} & {$\times$} \\ 
            \hline
            \begin{tabular}{@{}c@{}}\textbf{Blockchain}\\ \textbf{for Privacy}\\ \textbf{Preserving \cite{8874972}} \end{tabular} & {$\times$} & {\checkmark} & {\checkmark} & {$\times$} & {$\times$} & {$\times$} \\ 
            \hline
            \begin{tabular}{@{}c@{}}\textbf{Privacy}\\ \textbf{Preserving}\\ \textbf{IIoT \cite{jia2021blockchain}} \end{tabular} & {\checkmark} & {\checkmark} & {\checkmark} & {$\times$} & {$\times$} & {$\times$} \\ 
            \hline
            \begin{tabular}{@{}c@{}}\textbf{Blockchained}\\ \textbf{On-Device}\\ \textbf{FL \cite{8733825}} \end{tabular} & {\checkmark} & {$\times$} & {\checkmark} & {$\times$} & {\checkmark} & {$\times$} \\ 
            \hline
            \begin{tabular}{@{}c@{}}\textbf{Decentralized}\\ \textbf{aggregator}\\ \textbf{free FL \cite{ramanan2020baffle}} \end{tabular} & {\checkmark} & {$\times$} & {\checkmark} & {$\times$} & {\checkmark} & {$\times$} \\ 
            \hline
            \begin{tabular}{@{}c@{}}\textbf{BPFISH}\\ \textbf{(Proposed)} \end{tabular} & {\checkmark} & {\checkmark}  & {\checkmark} & {\checkmark} & {\checkmark} & {\checkmark}\\ 
            \hline
            
         \end{tabular}
     \end{center}
\end{table}

\subsection{Blockchain in Smart Healthcare}

In recent years, many research have proven that blockchain is a promising solution for ensuring the confidentiality, privacy preserving and distributed sharing of sensitive health information.
Oscar Novo \cite{novo2018blockchain} focused on distributed access management based on blockchain in IoT. However, the work does not consider the privacy of Electronic Medical Records (EMRs) of the patients. 
In \cite{dagher2018ancile, zhang2018blockchain}, blockchain-based EMRs access control model provided different access levels to different types of users, which had been controlled by the hospitals in traditional EMRs access control system.
Li Chaoyang et al. \cite{9261429} proposed healthchain, a peer-to-peer EMRs management and trading system based on a consortium blockchain. Healthchain system allows the access of patient EMRs in different institutions, and the EMRs can be traded between different patients.
These works mainly focused on healthcare record management but did not consider FL optimization considering privacy preservation.

\subsection{Association in Smart Healthcare} Many studies had been done for user association in resource offloading considering UAVs and IoTs enabled edge computing framework \cite{lim2021towards, 9112693}.
Smart user matching intersection with blockchain for stable matching in ultra-dense wireless networks had been proposed for computation offloading in \cite{9112693}. 
M. Kadadha et al. \cite{KADADHA2021103155} and J. An et al. \cite{8726129} explored Gale-Shapley based matching algorithm for node selection in blockchain-based crowdsourcing model.
R. Chen et al. \cite{9176324} presented a matching algorithm by considering the preferences from both patients and physicians to reduce the waiting time of patient.
However, the above stated works cannot be directly applied to the proposed framework due to ill posed nature of miners and MCs in blockchain and FL based healthcare architecture.

\subsection{FL, DP and Blockchain for Smart Healthcare}

C. Li et al. \cite{lu2019blockchain} and K. Gai et al. \cite{8874972} proposed privacy-preserving data sharing architecture based on blockchain for industrial Internet of Things by integrating FL and blockchain technology.
In \cite{jia2021blockchain}, privacy protection schemes for data in blockchain have been applied using data perturbation techniques like DP. However, these works do not consider incentives and the privacy preserving technique in smart healthcare domain.
In \cite{8733825} and \cite{ramanan2020baffle}, blockchain-based FL for privacy-preserving had been applied to eliminate the centralized global model aggregation.

The above existing works mainly focused on improving the performance of learning algorithms in FL and they did not provide incentives to the MCs participated in the FL process. However, it is necessary to provide incentive since participating institutes in FL requires data collection and model training on the device itself. 
Users are unlikely to participate in FL activities unless they are rewarded because model training consumes energy and requires a constant network connection. 
Moreover, none of the existing approaches has considered privacy-preserving decentralized FL with association and incentive mechanisms altogether as shown in Table \ref{comparisontable}. In this paper, we proposed a blockchain-based privacy-preserving FL combined with incentive and association mechanisms for smart healthcare applications.

\begin{figure}
    \centering
    \includegraphics[width=1\linewidth,height=7cm]{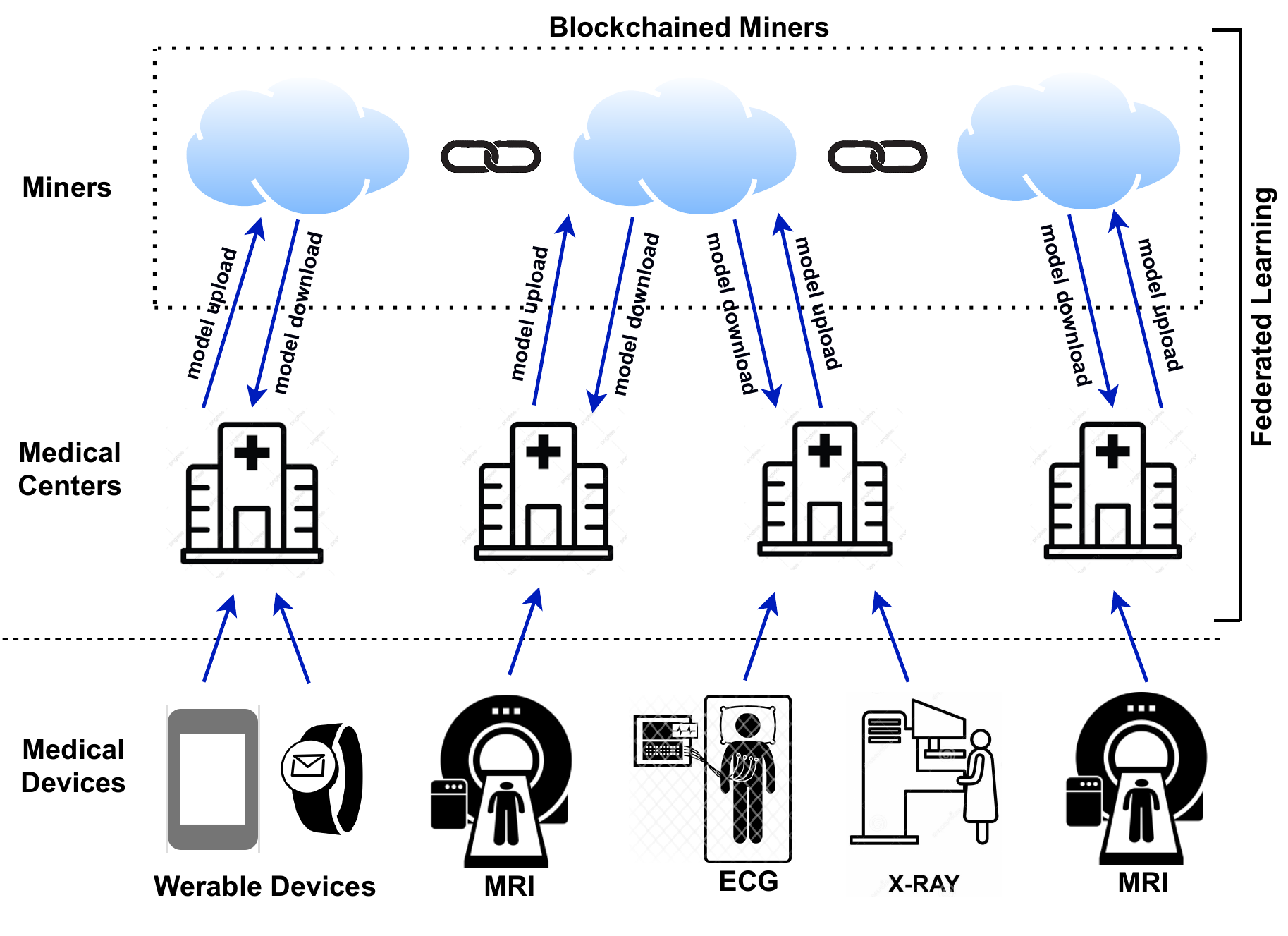}
    \caption{\label{application_model}System model of the proposed framework.}
\end{figure}

\section{System Model and Problem Formulation} \label{system}

As shown in Fig. \ref{application_model}, we consider a smart healthcare scenario consisting of $N$ MCs and $S$ miners represented by $\mathbb{C} = \{C_1, \cdots , C_n, \cdots,  C_{N}\}$ and $\mathbb{M} = \{M_1,\cdots, M_s, \cdots M_{S}\}$, respectively. Furthermore, let $X_n=\{x_n^1, \cdots, x_n^d, \cdots, x_n^{D_n}\}$ be the set of $D_n$ data samples available at MC $C_n \in \mathbb{C}$, collected from different patients admitted in the MC. MCs collaboratively train a shared model using their own data without sharing it to a central server. Specifically, each MC trains a model locally using their data and then send it to the miners for distributed aggregation \cite{jia2021blockchain}, resulting in a global model. 
An MC is characterized by a tuple $<$$X_{n}, f_{n}$$>$, where $X_{n}$ and $f_{n}$ represent the available amount of data samples for local model training and, the available number of CPU cycles at MC $C_n \in \mathbb{C}$, respectively. The above-stated scenario needs two-stage categorization as described in the following:

  \textbf{Stage 1: MMA:} In this stage, each MC finds an association with a miner. The association is formed by considering the utilities of both the MCs and miners using deferred acceptance algorithm (discussed in Section \ref{matching}).
  
  \textbf{Stage 2: Blockchain-based privacy-preserving FL:} MCs train local model and send model weights to the associated miner without sharing the healthcare data. Thus, hereby preserving the privacy of the patient's healthcare data. Miners add the models to the blockchain network. Blockchain provides distributed model weights sharing framework that is robust to the single point of failure (discussed in Section \ref{FL}).

To set up the association between MCs and miners (Stage 1), there is a need for preference order of one over the other, depending on maximum utilities; described in the following.

\subsection{Utility of Miner} \label{Utility of miners}
The blockchain network provides mining rewards for verification and adding blocks to the blockchain.
When a miner $M_s$ adds a block, its mining reward is provided from the blockchain network, as does in the traditional blockchain network \cite{article}. Each miner adds the local models from MCs to the blockchain. However, miner $M_s$ adds local model from $C_n$ as a block to the blockchain if and only if $M_s$ and $C_n$ are associated with each other. Specifically, the association between miner and MC is given as follows: 
\begin{equation}
  y_{n,s} = \begin{cases}
                1, & \text{if MC $C_n$ gets associated with miner $M_s$},\\
                0, & \text{otherwise}.
            \end{cases}
            \label{yns}
\end{equation}

Let $R_s$ be the revenue of the miner $M_s$ in the FL process and $\mathcal R$ be the mining reward for adding a block to the blockchain \cite{article}, which is the same for all the miners in the network. Each miner solves the PoW multiple times i.e., in every iteration. Therefore, the total revenue for a miner $M_s$, is represented as follows:
\begin{equation} \label{rev}
        R_{s} = T \cdot h(\mathcal R),
\end{equation}
where $h(\cdot)$ is a monotonically increasing function for $\mathcal R$ and $T$ is the number of global iterations. For simplicity, we consider following function to define $h(\mathcal R)$:
\begin{equation} \label{umi}
        h(\mathcal R) = \mathcal R \sum_{C_n \in \mathbb{C}} y_{n,s},
\end{equation}
where $\sum_{C_n \in \mathbb{C}} y_{n,s}$ is the number of blocks added to the blockchain by the miner $M_s$ in each iteration which is equivalent to the number of MCs associated with the miner $M_s$. Therefore, the total revenue of a miner can be rewritten as follows:
\begin{equation} \label{frev}
        R_{s} = T \mathcal R \sum_{C_n \in \mathbb{C}} y_{n,s}.
\end{equation}

The reward is an essential component of an FL based framework for the MC to participate in the FL process. The reward of MC is related to the number of data samples available for training at MC, $C_n$. We consider that the reward is linearly dependent on the number of data samples. Therefore, the reward offer to MC, $C_n$ from miner $M_s$ for training local model on data sample $D_n$, is calculated as follows \cite{8733825}:
\begin{equation} \label{reward}
     R_{n,s} = R_s \frac{D_n}{\sum_{C_n \in \mathbb{C}} D_n}.
\end{equation}

We assume that the utility and the total revenue of a miner are linearly dependent. Therefore, the utility of a miner is defined as follows:
\begin{equation} \label{fumi}
        U_{n,s}^{Min} = R_s - R_{n,s}.
\end{equation}

\begin{table}[t]
    \centering
        \caption{Symbol description}
        \begin{tabular}{|l|c|}
            \hline
            \textbf{Symbol} & \textbf{Description}\\
            \hline
            \hline
            $\mathbb{C}$ & Set of MCs\\
            N & Number of MCs \\
            $\mathbb{M}$ & Set of Miners \\
            $\textbf{w}^{(n)}$ & Model weights at MC $C_n$\\
            $X_n$ & Dataset at MC $C_n$\\
            $D_n$ & Number of data samples at $C_n$\\
            $l(\textbf{w}^{(n)}, x_n^d)$ & Loss per data sample\\
            $\alpha, B$ & Learning Rate and Batch Size\\
            T & Number of global iteration\\
            $I_n$ & Number of local iteration of MC, $C_n$ \\
            $G, G'$ & Computed Gradient and Clipped Gradient\\
            $G''$ & Gradient after adding noise\\
            $A$ & Gradient bound\\
            $\mathcal{N}, \sigma$ & Gaussian Noise and Noise scale\\
            $\epsilon, \delta$ & Privacy parameters\\
            $\textbf{w}^{(n,t)}$ & Weights at MC, $C_n$ at global iteration $t$\\
            $\textbf{W}^{(t)}$ & Aggregate Weights at global iteration $t$\\
            $\mathbb{K}$ & Set of MCs associated with miners\\
            $|\mathbb{K}|$ & Number of MCs participated in FL process\\
            $y_{n,s}$ & Association between MC and miner\\
            $\mathcal{R}, R_{n,s}$ & Mining reward and miner revenue\\
            $\varphi$ & Cost per unit energy\\
            $U_{n,s}^{Min}, U_{n,s}^{MC}$ & Utility of miner and MC\\
            $f_n, \beta_n$ & CPU cycles/second, CPU cycles/sample\\
            $I_n$ & Number of local iteration\\
            $\mu_n$ & Transmit power\\
            $E_{n,s}^{Trans}$ & Energy consumption in model upload \\
            $E_{n}^{Comp}$ & Energy consumption in model computation\\
            $\tau_{n,s}^{Trans}$ & Model computation time at MC\\
            $\tau_{n}^{Comp}$ & Transmission time from MC to miner\\
            $\tau^{th}$ & Threshold\\
            $\Upsilon_{n}$ & Utility list of $C_n$ in increasing order\\
            $\Delta_s$ & Utility list of $M_s$ in increasing order\\
            $\rho, \eta$ & Weights assigned to utility and FL loss\\
            $\mathbb{A}_s, \mathbb{B}_s$ & Feasible association candidates of $M_s$ and $C_n$ \\
            $P_n$ & Preference list of MC $C_n$\\
            $P_s$ & Preference list of miner  $M_s$\\
            
            \hline
        \end{tabular}
    \label{symboltable}
\end{table}

\subsection{Utility of Medical Center} \label{Utimc}

We use the utility as a criterion to determine the preferences of the MCs and miners for the association. Since MC trains model in the FL process with the consumption of energy, we consider computation energy while defining the utility of an MC.
Let $f_{n}$ be the number of CPU cycles per second (computation capacity) of MC $C_n \in \mathbb{C}$. Let $\beta_n$ (cycles/sample) be the number of CPU cycles needed for computing one sample data $x_n^d \in X_n$ at MC, $C_n$.
The energy consumption to calculate the total number of $\beta_n D_{n}$ CPU cycles at MC, $C_n$ can be derived as \cite{9264742}:
\begin{equation} \label{ec1}
        E_{n} = \kappa \beta_n D_{n} f_{n}^2,
\end{equation}
where $\kappa$ is a coefficient that depends on the chip architecture. In the FL process, an MC, $C_n$ requires to compute $\beta_n D_{n}$ CPU cycles in each $I_{n}$ local iterations. Therefore the total energy consumption in local model training is given by:
\begin{equation} \label{ec}
        E_{n}^{comp} = I_{n} E_{n} = \kappa I_{n} \beta_n D_{n} f_{n}^2.
\end{equation}
Moreover, the computation time taken by an MC, $C_n$ to train local model can be defined as follows: $\tau_{n}^{comp} = I_{n} \beta_n D_n / f_{n}$.

Following the local model training, model is sent from the MC to the associated miner. For simplicity, we define the achievable data rate between MC, $C_n$  and miner $M_s$ as: $\upsilon_{n,s} = \displaystyle  Q V \log_2\left(1 + SINR_{n,s} \right)$. Here, $V$ is the total number of allocated PRBs\footnote{PRB is the smallest unit that can be assigned to a device in the 5G network \cite{9371426}.} between MC and miner. $SINR_{n,s}$ and $Q$ are the signal-to-interference-plus-noise ratio (SINR) \footnote{For simplicity, we assume that the channel exhibits flat fading. However, this can be easily extended to frequency-selective fading channels as well \cite{9371426}.} and the bandwidth of PRB allocated between MC and miner, respectively. The total transmission time\footnote{We assume that model download time between miners and MCs is negligible compared with the transmission time as usually the downlink bandwidth is significantly larger than the uplink bandwidth \cite{kang2019incentive}.} \cite{kang2019incentive} for an MC in uploading the model weights of size $H$ to the miner during the FL process is as follows: $\tau_{n,s}^{trans} = T  H / \upsilon_{n,s}$. Given the same model across all MCs, the size of the local model remains constant, independent of the number of iterations or the amount of data available. The energy consumption in the transmission of the trained local model can be defined as follows \cite{lim2021towards}:
\begin{equation} \label{et}
        \zeta_{n,s} = E_{n,s}^{trans} = \tau_{n,s}^{trans} \mu_{n}.
\end{equation}
where, $\mu_{n}$ is the transmit power of MC, $C_n$.

Therefore, the utility, $U_{n,s}^{MC}$ of an MC, $C_n$ can be defined as follows:
\begin{equation} \label{umc}
    \begin{split}
        U_{n,s}^{MC} & = R_{n,s} - \varphi (E_{n}^{comp} + E_{n,s}^{trans}), \\
        & = R_{n,s} - \varphi (\kappa I_{n} \beta_{n} D_{n} f_{n}^2 + \zeta_{n,s}),
    \end{split}
\end{equation}
where $\varphi$ represents the cost per unit energy.

During the FL process, each MC has to upload the local model within a specified time. Particularly, to conduct FL process reliably and to find the feasible association candidates between miners and MCs, the time constraint needs to satisfy:
\begin{equation} \label{timecons}
        \tau_{n}^{comp} + \tau_{n,s}^{trans} \leq \tau_{th}, \forall C_n \in \mathbb{C}, \forall M_s \in \mathbb{M},
\end{equation}
where $\tau_{th}$ is the predefined threshold.

\subsection{Federated Learning}

We define a vector $\textbf{w}^{(n)}$ as weights related to the FL model for MC, $C_n$. We introduce the loss function $l(\textbf{w}^{(n)}, x_n^d)$ of the model, which indicates the FL performance over an input sample. The loss function varies depending on the learning problem. Different learning problems use different loss functions. 
In smart healthcare, loss defines how far the model's prediction or the outcomes is from the actual outcomes of the health condition \cite{10.1145/3410566.3410598}. The FL local model training problem at MC $C_n \in \mathbb{C}$ for miner $M_s$, using its own data $X_n$ can be formulated as follows \cite{reisizadeh2020fedpaq}:
\begin{equation} \label{learningproblem}
    \min_{\textbf{w}^{(n)}} \mathcal{L}_{n,s}(\textbf{w}^{(n)}) \min_{\textbf{w}^{(n)}} \frac{1}{D_n} \sum_{x_n^d \in X_n} l(\textbf{w}^{(n)}, x_n^d).
\end{equation}

In order to get better privacy, we integrate DP in the local model training. DP ensures that the output of a differentially private algorithm gives the same output with $\epsilon$ error whether or not a local data sample is included in the input of the algorithm. $\epsilon$ is the privacy budget that is bound on the loss of the privacy of the algorithm. 
We use the variant of DP definition introduced in \cite{Abadi_2016}, that satisfies ($\epsilon, \delta$)-DP, where $\delta$ defines the bound that the privacy guarantee does not hold (which is preferably very small positive number).
\begin{definition}
     A randomized algorithm $\mathcal Z:\mathcal X\rightarrow \mathcal Y$ with domain $\mathcal{X}$ and range $\mathcal{Y}$ satisfies ($\epsilon, \delta$)-DP if for any two adjacent datasets D' and D" that differ by one data record and for any subset $O \subseteq \mathcal Y$, the following probability condition holds:
\begin{equation}
    Pr[\mathcal{Z}(D')\in O] \leq  e^{\epsilon} Pr[\mathcal{Z}(D'')\in O] +\delta.
\end{equation}
\end{definition}
Privacy can be achieved in the algorithm by adding random noise to the gradient computation so that the resultant model is noisy.
Therefore, we update the weights at $i$-th ($i \in [1, I_n])$ local iteration for optimizing the loss function as follows \cite{Abadi_2016}:
\begin{equation} \label{updateproblem}
    \textbf{w}^{(n,i+1)} = \textbf{w}^{(n,i)} - \alpha G'',
\end{equation}
where $\alpha$ is the learning rate of the model. Moreover, $G''$ is calculated using Eq. \eqref{noisegradproblem}:
\begin{equation} \label{noisegradproblem}
    G'' = G' + \mathcal{N}(0, \sigma^2 A^2\mathbb{I}),
\end{equation}
where $\mathcal{N}(0, \sigma^2 A^2\mathbb{I})$ is the Gaussian noise, $\mathbb{I}$ is the identity matrix and $G'$ is the clipped gradient using $L_2$ norm of the gradient calculated as follows:
\begin{equation} \label{clipgradproblem}
    G' = \Delta  \mathcal{L}_{n,s}(\textbf{w}^{(n)})/ \max \left(1, \frac{||\Delta  \mathcal{L}_{n,s}(\textbf{w}^{(n)})||_2}{A}\right),
\end{equation}
where $\sigma = \sqrt{2\log\frac{1.25}{\delta}} / \epsilon$ for $\epsilon > 0$ and $A$ is the gradient bound.
DP requires bounding the impact of each data sample on $G''$. So, each gradient is clipped in the $L 2$ norm, i.e., the gradient $\Delta  \mathcal{L}_{n,s}(\textbf{w}^{(n)})$ is replaced by $\Delta  \mathcal{L}_{n,s}(\textbf{w}^{(n)})/ \max \left(1, \frac{||\Delta  \mathcal{L}_{n,s}(\textbf{w}^{(n)})||_2}{A}\right)$, for a gradient bound $A$. This clipping ensures that if $||\Delta  \mathcal{L}_{n,s}(\textbf{w}^{(n)})||_2 \leq A$, then $\Delta  \mathcal{L}_{n,s}(\textbf{w}^{(n)})$ is preserved, whereas if $||\Delta  \mathcal{L}_{n,s}(\textbf{w}^{(n)})||_2 > A$, gradient gets scaled down to be of norm $A$. This gives the global minimization problem over the collection of all the loss functions from MCs. This global minimization problem is minimized by finding the optimal weights, $\textbf{W}^{T}$ and is defined as follows:
\begin{equation} \label{wcost}
    \textbf{W}^{T} = \min_{\textbf{W}} J(\textbf{W}),
\end{equation}
where $J(\textbf{W})$ is the total loss over the collection of all the MCs, given as follows:
\begin{equation} \label{cost}
    J(\textbf{W}) = \frac{1}{\sum_{C_n \in \mathbb{C}} \sum_{M_s \in \mathbb{M}} y_{n,s}} \sum_{C_n \in \mathbb{C}} \sum_{M_s \in \mathbb{M}} y_{n,s} \mathcal{L}_{n,s}(\textbf{w}^{(n)}),
\end{equation}
where $\textbf{W}$ is the weights of the global model. $\sum_{C_n \in \mathbb{C}, M_s \in \mathbb{M}} y_{n,s}$ gives the total number of MCs participated in the FL process. Moreover the descriptions of used
symbols are given in Table \ref{symboltable}.

\subsection{Problem Formulation}


We formulate a joint optimization problem having goal to maximize the utility of both the miner and MC and minimize the FL loss function while satisfying the FL process delay requirement. Utility of MC involves determining the miner associated with, amount of data present at MC for local model training, and the uplink transmit power of each MC for model update transmission to the miner. The utility of miner involves mining rewards for adding blocks to the blockchain. Therefore, the total utilities of miners and MCs is given by:
\begin{equation} \label{tuti}
    U = \sum_{C_n \in \mathbb{C}} \sum_{M_s \in \mathbb{M}} y_{n,s}(U_{n,s}^{Min} + U_{n,s}^{MC}),
\end{equation}
where $U_{n,s}^{Min}$ and $U_{n,s}^{MC}$ are defined in the above Eq. \eqref{fumi} and Eq. \eqref{umc}, respectively.

The minimization problem of Eq. \eqref{wcost} is equivalent to the following maximization problem:
\begin{equation} \label{maxcost}
    \textbf{W}^{T} = \max_{\textbf{W}} \{-J(\textbf{W})\}.
\end{equation}
As a result, we combine both the factors i.e., utilities of miners and MCs, and FL process loss function as follows:
\begin{equation} \label{prblemcombined}
    F = \rho U + \eta \{-J(\textbf{W})\},
\end{equation}
where $\rho$ and $\eta$ are the weights assigned to the utility and the FL loss function and $\rho + \eta = 1$.

Therefore, the optimization problem of the system model can be formulated as follows: 
\begin{equation}\label{Prob:main}
\begin{array}{ll}
\displaystyle \max_{y_{n,s}, \textbf{W}} & F
\\[2ex]
\text{S. T.} & \text{1.1. }  \displaystyle U_{n,s}^{MC} > 0, U_{n,s}^{Min} > 0, \forall C_n \in \mathbb{C}, \forall M_s \in \mathbb{M},
\\[2ex]
& \text{1.2. } \displaystyle \sum_{C_n \in \mathbb{C}} \sum_{M_s \in \mathbb{M}} y_{n,s} > 0,
\\[2ex]
& \text{1.3. }  \displaystyle \text{Eqs.} \hspace{1mm} \eqref{yns} \hspace{1mm} \text{and} \hspace{1mm} \eqref{timecons},
\end{array}
\tag{P1}
\end{equation}
$\forall C_n \in \mathbb{C}, \forall M_s \in \mathbb{M}$.
Constraint 1.1 tells that utility of MCs and miners should be greater than zero.
Constraint 1.2 tells that at least one association between miners and MCs is possible in the FL process. Constraint 1.3 is described in the above Eq. \eqref{yns} and Eq. \eqref{timecons}, respectively. 

The problem given in \ref{Prob:main} is challenging since it involves the optimization of utilities as well as FL loss values. We decouple this problem into the optimization of utility and optimization of FL loss and solve them separately. Firstly, the optimization of utility is solved for the Miner-MC Association (MMA) by applying Algorithm \ref{algo.matching}. Then the FL loss optimization is solved using Stochastic Gradient Descent (SGD) algorithm with DP and blockchain technology as described in Algorithm \ref{fdBlock}.

\section{Miner-MC Association} \label{matching}

Associations between miners and MCs are formed in coordination with a trusted entity as shown in Fig. \ref{flowchart}. To perform the association, MCs and miners submit their computational power and data size, and amount of mining reward to the trusted entity, respectively. Here, computational power refers to the CPU cycles of each MC. Associations between miners and MCs are formally defined as follows:

\begin{definition}MMA: An association between miner $M_s \in \mathbb{M}$ and MC $C_n \in \mathbb{C}$ is a mapping $g:\mathbb{C} \rightarrow \mathbb{M}$, such that $g(C_n) = M_{s}, \forall C_n \in \mathbb{C}, \forall M_s \in \mathbb{M}$.
\label{asssociation}

\end{definition}
However, performing MMA in practice is subject to some constraints. To perform the FL process reliably, all the constraints 1.1, 1.2 and 1.3 need to satisfy. Thus, we define a feasible MMA as follows:
\begin{definition} Feasible MMA: An MMA is feasible, if:
\label{feasible}
    \begin{itemize}
        \item $\forall C_n \in \mathbb{C}$, each of its association with miner, $g(C_n) = M_s$ should satisfy Eq. \eqref{timecons}
        
        \item $\forall C_n \in \mathbb{C}$, there exists at most one association with miner, i.e., $|\{g(C_n)\}| \leq 1$.
    \end{itemize}
\end{definition}

The above Definition \ref{asssociation} implies that $g$ is a many-to-one association i.e., $g(C_n)$ is not unique. The interpretation of $g(C_n) = \phi$ implies that for some $C_n \in \mathbb{C}$, there is no association due to non satisfactory of constraints. 
The result of the function determines the successfully associated MC and the association between the miner and MC, e.g., $\mathbb{K}$ and $g \equiv \mathbb{Y}$, where
\begin{equation} \label{kasso}
    \mathbb{K} = \{C_n|y_{n,s} = 1, \forall C_n \in \mathbb{C}, \forall M_s \in \mathbb{M} \}
\end{equation}
and 
\begin{equation} \label{ynsasso}
    \mathbb{Y} = \{y_{n,s}|y_{n,s} = 1, \forall C_n \in \mathbb{C}, \forall M_s \in \mathbb{M} \}.
\end{equation}

\begin{algorithm}[!t] 
	\begin{algorithmic}[1]
		\item[]\hspace{-0.6 cm} \textbf{Input:} Set of MCs $\mathbb{C}$, set of miners $\mathbb{M}$, number of data samples $D_n$, CPU cycles $f_n$, cycles per sample $\beta_n$, threshold $\tau_{th}$.
		
		\item[]
	   \State \textbf{Initialization:} $U_{n,s}^{Min}$ and $U_{n,s}^{MC}, \forall C_n \in \mathbb{C}, \forall M_s \in \mathbb{M}$ 
        \State \textbf{while} $(|\mathbb{C}| \neq 0)  $ \textbf{do}
        
            \State \hspace{0.30 cm} Find $\mathbb{A}_{s} = \{C_n|(\tau_{n}^{comp} + \tau_{n,s}^{trans}) \leq \tau_{th}, \forall C_n \in \mathbb{C}\}$ and 
            \item[] \hspace{0.30 cm} $\mathbb{B}_{n} = \{M_s|U_{n,s}^{Min} > 0, \forall M_s \in \mathbb{M}\}$ 
	   
    	    \State \hspace{0.30 cm} Arrange the utility lists in non increasing order
    	    \item[] \hspace{0.30 cm} $\Upsilon_{n} = desc\{U_{n,s}^{Min}|\forall M_s \in \mathbb{B}_n\}$ and 
    	    \item[] \hspace{0.30 cm} $\Delta_{s} = desc\{U_{n,s}^{MC}|\forall C_n \in \mathbb{A}_{s}\}$, 
    		
    	    \State \hspace{0.30 cm} Create preference lists $P_{n}$ and $P_{s}$ using $\Upsilon_{n}$ and $\Delta_{s}$
            \State \hspace{0.30 cm} $\mathbb{W}_{s} \leftarrow \phi, Rej_s \leftarrow \phi, Rej_n \leftarrow \phi, \forall C_n \in \mathbb{C}, \forall M_s \in \mathbb{M}$
    		\State \hspace{0.30 cm} \textbf{for all} $C_n \in \mathbb{C}$ \textbf{do}
    		\State \hspace{0.70 cm} Find the miner $M_s$ in $P_{n}$ with the highest utility,
    		\item[] \hspace{0.70 cm} such that the utility of which is $U_{n,s*}^{Min}$
    		\State \hspace{0.70 cm} \textbf{if} $M_{s*} = 0$ \textbf{then}
    		    \State \hspace{1.10 cm} $y_{n,s} = 0,$ $\forall M_s \in P_{n}$

    		\State \hspace{0.70 cm} \textbf{else}
    		    \State \hspace{1.10 cm} $C_n$ applies to $M_s$
                \State \hspace{1.10 cm} $M_s$ removes $C_n$ from $P_{n}$ and add it into $\mathbb{W}_{s}$
    		\State \hspace{0.70 cm} \textbf{end if}
        	\State \hspace{0.30 cm} \textbf{end for}
            \item[]
            \State \hspace{0.30 cm} \textbf{for all} $M_s \in \mathbb{M}$ \textbf{do}
    		\State \hspace{0.70 cm} Find the MC, $C_{n*}$ in $\mathbb{W}_{s}$ according to $P_{s}$ with the
    		\item[] \hspace{0.70 cm} highest utility
    		\State \hspace{0.70 cm} $g(C_{n*}) = M_s, y_{{n*},s} = 1$, $\mathbb{Y} \leftarrow \mathbb{Y} \cup  \{y_{n*,s}\}$
    		\State \hspace{0.70 cm} $y_{n,s} = 0$, $\forall C_{n} \in \mathbb{W}_{s}$, $C_n \neq C_{n*}$
            \State \hspace{0.70 cm} Remove MCs in $\mathbb{W}_{s}$ (except $C_{n*}$) from $P_{s}$ into $Rej_{s}$
    		\State \hspace{0.70 cm} Rejected MCs in $Rej_{s}$ adds $M_s$ from $P_{n}$ into $Rej_{n}$
    		\State \hspace{0.70 cm} $\mathbb{K} \leftarrow \mathbb{K} \cup  \{C_{n*}\}$, $\mathbb{C} \leftarrow \mathbb{C} - \{C_{n*}\}$
    		\State \hspace{0.30 cm} \textbf{end for}
    		\item[]
    		\State \hspace{0.30 cm} \textbf{if} No MCs are rejected \textbf{then}
    		    \State \hspace{0.70 cm} Break
    		\State \hspace{0.30 cm} \textbf{end if}
    		
	   \State \textbf{end while}
	   \item[]
	   \textbf{Output:} Association between miners and MCs i.e., $\mathbb{Y}$ and successfully associated MCs i.e., $\mathbb{K}$.
		\caption{Proposed MMA Algorithm} 
		\label{algo.matching}
	\end{algorithmic}
\end{algorithm}

\subsection{Proposed MMA Algorithm}

The proposed MMA is based on the deferred acceptance algorithm \cite{gale1962college}. We modify the deferred acceptance algorithm in particular due to its feasibility and stability, allowing the proposed MMA algorithm to find a stable association result as described in the following.

In the beginning, miners and MCs select feasible association candidates that meet the criteria, i.e., $(\tau_{n}^{comp} + \tau_{n,s}^{trans}) \leq \tau_{th}$
as defined in Definition \ref{feasible} and $U_{n,s}^{Min} > 0$, respectively. Miner $M_s$'s and MC $C_n$'s feasible association candidates are denoted by $\mathbb{A}_{s}$ and $\mathbb{B}_{n}$, respectively (line 3). We arrange the utility lists of MC, $C_n$ and miner, $M_s$ in non increasing order, denoted by $\Upsilon_{n} = desc\{U_{n,s}^{Min}|\forall M_s \in \mathbb{B}_n\}$ and $\Delta_{s} = desc\{U_{n,s}^{MC}|\forall C_n \in \mathbb{A}_{s}\}$, respectively (line 4), where $desc\{\cdot\}$ is to represent the ordered list in non-increasing order. Miners and MCs then construct their preference lists by arranging the utilities in non-increasing order based on the feasible candidate list (line 5). The preference lists of miners and MCs present in $\Upsilon_{n}$ and $\Delta_{s}$ are given by $P_{n}$ and $P_{s}$, $\forall C_n \in \mathbb{C}, \forall M_s \in \mathbb{M}$. At the start of each association iteration, all waiting lists, $\mathbb{W}_s$ and miner $M_s$'s and MC, $C_n$'s rejection lists denoted by $Rej_{s}$ and $Rej_{n}$, respectively are initialized as empty lists (line 6) i.e., $\mathbb{W}_s \leftarrow \phi$, $Rej_{s} \leftarrow \phi$, and $Rej_{n} \leftarrow \phi$,$\forall M_s \in \mathbb{M}, \forall C_n \in \mathbb{C}$.

Once the preference lists of miners and MCs are set up, the main matching algorithm between miners and MCs begin as shown in Algorithm \ref{algo.matching}. At each iteration, $P_{s}$, $P_{n}$, $Rej_{s}$, $Rej_{n}$ and $y_{n, s}$, $\forall C_n \in \mathbb{C}, \forall M_s \in \mathbb{M}$, are updated. Furthermore, MC, $C_n$ selects its preferred $M_s$ according to the preference list (line 8). Then MC $C_n$ applies to its preferred miner $M_{i*}$ in $P_{s}$ for association, i.e., $U_{n,s*}^{Min} > U_{n,s}^{Min}, \forall M_s \in P_{n}, M_{s} \neq M_{s*}$ (line 12). Every miner $M_s \in \mathbb{M}$ receives a set of request for association from MCs. These requests are added to $\mathbb{W}_{s}$, called a waiting list, $\forall M_s \in  \mathbb{M}$ (line 13). If MC, $C_n$ has the greatest utility in waiting list $\mathbb{W}_{s}$ of miner $M_s$, miner $M_s$ accepts the MC and rejects the other MCs in $\mathbb{W}_{s}$ (line 17). The association between miner and MC is updated as $y_{n,s}$ (line 18).
Miners remove all the MCs except for $C_{n*}$ from its preference list and add it into the rejected list (line 20). 
The rejected MCs in $Rej_{s}$ remove the miners who reject them from their preference lists and add them to their rejected lists (line 21). 
All MCs that are associated with a miner are added to set $\mathbb{K}$ and removed from the set of MCs, $\mathbb{C}$ (line 22). 
If no MCs are rejected i.e., all the MCs are associated with one of the miners, then this completes MMA algorithm (lines 24-26). 
Otherwise, the above-mentioned procedure is repeated until no more MCs are rejected (line 2-27).

\subsection{Theoretical Analysis of MMA Algorithm}

The stability of the association algorithm implies the robustness to change in the association that increases the utilities of miners and MCs. If the association is unstable, an MC is willing to change its associated miner if favourable to the MC. Such a network with unstable associations eventually results in an unsatisfactory and unreliable association. For a stable MMA, the condition of the nonexistence of blocking pairs must satisfy. We formally define the blocking pairs as follows.
\begin{definition} Blocking Pair: For every MC, $C_n$, a pair $(C_n, M_s)$ is defined as blocking pair if all the following conditions are satisfy:
    \begin{itemize}
        \item $C_n$ associated with miner $M_{s*} \in \mathbb{M}$.
        
        \item There exists another pair $(C_n, M_s)$, such that $U_{n,s}^{MC} > U_{n,s*}^{MC}, M_{s} \neq M_{s*}, M_s \in \mathbb{M}$.
        
    \end{itemize}
    \label{BP}
\end{definition}

The higher utility can be obtained by blocking pairs. This indicates that the MC has a strong desire to change the association, implying that the association is unstable. Based on the definition of blocking pair, we define the stability of the MMA in the following. 

\begin{definition} Stability:
     An MMA is stable if it satisfies the condition of nonexistence of blocking pairs.
\end{definition}

As a result, we can make Lemma \ref{prepo1} concerning the stability of the MMA formed by applying Algorithm \ref{algo.matching}.

\newtheorem{prop}{Lemma}
\begin{prop}\label{prepo1}
Algorithm \ref{algo.matching} gives stable association.
\end{prop}

\begin{proof}
To prove the stability of the MMA algorithm, we show that there exists no blocking pair. Let there be a blocking pair $(C_n, M_q)$ for MC, $C_n$ after associated with miner $M_r$, $ \forall C_n \in \mathbb{C}, M_q,M_r \in \mathbb{M}$.
According to Definition \ref{BP}, $(C_n, M_q)$ satisfies $U_{n,q}^{MC} > U_{n,r}^{MC}$. According to Algorithm \ref{algo.matching}, if MC, $C_n$ is associated with miner $M_r$, then $C_n$ associated with $M_r$ should have a greater utility than that the utility obtained by association with other miners. However, $C_n$ fails to form association with $M_q$ if and only if the utility of $C_n$ is higher if associated with miner $M_r$ than the utility obtained from the association with miner $M_q$. This contradicts with $(C_n, M_q)$, i.e., $U_{n,q}^{MC} > U_{n,r}^{MC}$. As a result, there exists no blocking pair after forming the association between MCs and miners by applying Algorithm  \ref{algo.matching}, which thus proves the Lemma.
\end{proof}

\begin{theorem}\label{theoremcompu}
Time complexity of MMA algorithm is
$O(N^{2}S\log NS)$.
\end{theorem}

\begin{proof} 
Algorithm \ref{algo.matching} iteratively formed the associations between miners and MCs. We can see from the algorithm that the number of iterations depends on the number of MCs to be associated with miners. Let consider the worst-case scenario where an MC is associated with a miner after both MC and miner search all their association candidates.
The total association candidate in $\mathbb{A}_{s}$ and $\mathbb{B}_{n}$ could be $2NS$. 
Therefore, line 3 takes $O(NS)$ time.
Construction of utility lists $\Upsilon_{n}$ and $\Delta_{s}$ (line 4) take $O(NS\log S)$ and $O(NS\log N)$ using Heap sort algorithm, respectively.
Creation of $P_n$ and $P_s$ (line 5) take $O(NS)$.
Selection of candidate miners for each MC (lines 7-15) takes $N$ times. Thus, the worst-case computational complexity of lines 7-15 is $O(N)$. Similarly, the worst-case computational complexity of lines 16-23 is $O(S)$. Lines 2-27 iterates over the number of MCs i.e., $N$. Therefore, the worst-case computational time complexity of MMA algorithm is 
$O(N(NS + NS\log S + NS \log N + N + S))$ i.e.,
$O(N^{2}S\log NS)$.
\end{proof}

This completes Stage 1 i.e., the association between miners and MCs. After the association, the actual privacy-preserving FL process begins. In the next section,
Stage 2 i.e., Blockchain and Privacy-preserving FL Inspired Smart Healthcare (BPFISH) is discussed in detail.

\begin{algorithm}[!t]
	\begin{algorithmic}[1]
		\item[]\hspace{-0.6 cm} \textbf{Input:} Training samples $X_n$, learning rate $\alpha$, batch size $B$, loss function $\mathcal{L}_{n,s}(\textbf{w}^{(n)}) = \frac{1}{D_n} \sum_{x_n^d \in X_n} l(\textbf{w}^{(n)}, x_n^d)$, number of global iteration $T$, noise scale $\sigma$, Gradient bound $A$.
		\item[]
		\\
		Initialized weights $\textbf{w}^{(0)}, \textbf{w}^{(n)} \leftarrow \textbf{w}^{(0)}$
		\State \textbf{for} each global iteration $t = 1, 2, \cdots, T$ \textbf{do}
		\item[]\hspace{0.30 cm} \textbf{Local Model Training} 
		
		\State \hspace{0.30 cm} \textbf{for} each $C_n \in \mathbb{K}$ \textbf{do}     \hspace{1 cm} \textbackslash \textbackslash Perform in parallel
        \State \hspace{0.70 cm} \textbf{for} each local iteration  $i = 1, 2, \cdots, I_n$ \textbf{do}
    		\State \hspace{1.10 cm} \textbf{for} each batch $b \in X_n$ \textbf{do}
    		    \State \hspace{1.50 cm} \textbf{for} each sample $x \in b$ \textbf{do}
            		\State \hspace{1.90 cm} Compute Loss $\mathcal{L}_{n,s}(\textbf{w}^{(n, t)})$
            		\State \hspace{1.90 cm} Compute gradient $G = \Delta  \mathcal{L}_{n,s}(\textbf{w}^{(n, t)})$
            		\State \hspace{1.90 cm} Gradient clipped $G' = G/ \max \left(1, \frac{||G||_2}{A}\right)$
        		\State \hspace{1.50 cm} \textbf{end for}
        		\State \hspace{1.50 cm} Add noise to the gradient $G'' = \frac{1}{B}(\sum_{x \in b} G'$ 
        		\item[] \hspace{1.50 cm} + $\mathcal{N}(0, \sigma^2 A^2\mathbb{I}))$
            	\State \hspace{1.50 cm} Updates local weights $\textbf{w}^{(n,t)} = \textbf{w}^{(n,t)} - \alpha  G''$
            	\State \hspace{1.50 cm} Upload $\textbf{w}^{(n,t)}$ to associated miner
            \State \hspace{1.10 cm} \textbf{end for}
	   	\State \hspace{0.70 cm} \textbf{end for}
		\State \hspace{0.30 cm} \textbf{end for}
		
		\item[]
		\item[]\hspace{0.30 cm} \textbf{Blockchained Miners}
    		\State \hspace{0.30 cm} \textbf{for}  each Miner $M_s \in \mathbb{M}$ \textbf{do}
        		\State \hspace{0.70 cm} Verify received $\textbf{w}^{(n,t)}$ and broadcast to all miners
        		\State \hspace{0.70 cm} Add $\textbf{w}^{(n,t)}$ to candidate block
        		\State \hspace{0.70 cm} Performs PoW and adds candidate block to
        		\item[] \hspace{0.70 cm} blockchain
        		\State \hspace{0.70 cm} Broadcast the new block to all miners
    		\State \hspace{0.30 cm} \textbf{end for}
	    
	        \item[]
    	    \item[] \hspace{0.30 cm} \textbf{Model Aggregation}
    	    \State \hspace{0.30 cm} Download local model weights from the associated 
    	    \item[] \hspace{0.30 cm} miner for $t$-th global iteration
    	    \State \hspace{0.30 cm} MCs Update the global weights as
    	    \item[] \hspace{0.30 cm} $\textbf{W}^{(t+1)} = \hspace{0.30 cm} \sum_{n=1}^{|\mathbb{K}|} p_n \textbf{w}^{(n,t)}$
    	    \State \hspace{0.30 cm} Initialize local weights as $\textbf{w}^{(n,t+1)} \leftarrow \textbf{W}^{(t+1)}$
    	    
    	    \State \textbf{end for}
    	    \vspace{0.2 cm}
    	    \item[]\hspace{-0.6 cm} \textbf{Output:} Optimal model weights, $\textbf{W}^{T}$.
	    
		\caption{Blockchain-based Privacy-preserving FL} 
		\label{fdBlock}
	\end{algorithmic}
\end{algorithm}

\begin{figure}[t!]
    \centering
    \includegraphics[width=1\linewidth]{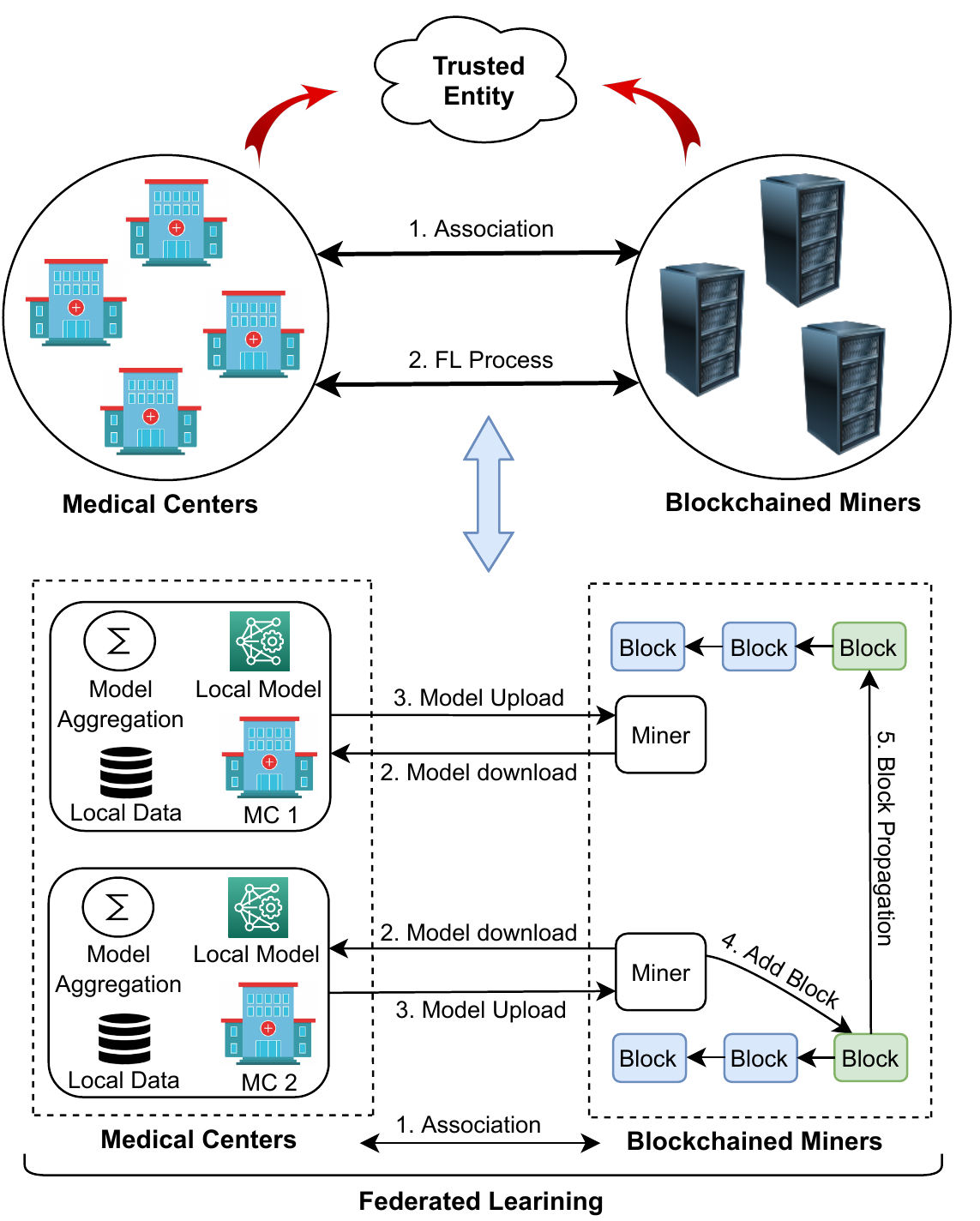}
    \caption{\label{flowchart}The proposed BPFISH framework.}
\end{figure}

\section{Proposed BPFISH Algorithm} \label{FL}

FL allows the collaborative learning of multiple MCs without sharing their local data. An adversary can use an inference attack to recreate the training data from the shared model. Therefore, we incorporate the concept of DP for privacy-preserving in model training. This can be achieved in the algorithm by adding random noise to the gradient computation. Moreover, we leverage the advantages of blockchain in this work. 
Blockchain provides an immutable and decentralized framework for model weights sharing. The proposed algorithmic steps are shown in the Fig. \ref{flowchart}. Specifically, FL process consists of two components: Blockchain for distributed local model weights sharing and local model training by MCs using their private data. Each MC trains the local model using their local data and updates the model weights to the miners. The FL local model training problem is to optimize the loss as given in Eq. \eqref{learningproblem}. The FL loss optimization is solved using Stochastic Gradient Descent (SGD) algorithm under DP as described in Algorithm \ref{fdBlock}.

In the beginning, blockchain adds the model that will be utilised by all MCs, as a genesis block on the blockchain. 
At the initialization step, each MC downloads the model weights present in the genesis block from its associated miner and initialized the weights for local model training (line 1). At each local iteration of the SGD, we computed the gradient $G$ for a batch of data samples and clipped each computed gradient (lines 6-10). Then add Gaussian noise drawn from Gaussian distribution to the average gradient and upload the updated weights to the associated miner (lines 11-13). This process repeats for each batch in every local iteration (lines 4-15). Each MC in the FL process performs lines 3-16 in parallel for every global iteration.

Miners add the local models from the associated MCs to the blockchain for distributed model weights sharing and provide an immutable framework. Miners verify the received weights and add them to the candidate block (lines 18-19). Miners perform PoW to add candidate block in the blockchain and broadcast the new block to all miners (lines 20-21). For model aggregation, MCs download local model from the associated miners and aggregate the model weights (lines 23-25) on-device in every global iteration. Each MC aggregates $|\mathbb{K}|$ local weights $\textbf{w}^{(n,t)}$ downloaded from its associated miner at $t$-th global iteration as follows:
\begin{equation} \label{eq2}
    \textbf{W}^{(t+1)} = \sum_{n=1}^{|\mathbb{K}|} p_n \textbf{w}^{(n,t)},
\end{equation}
where $p_n = \frac{D_n}{D}$ is the weightage given to each MC, $C_n$. $\textbf{w}^{(n,t)}$ is the local model weights at MC, $C_n$ at global iteration $t$. Here $D = \sum_{n=1}^{|\mathbb{K}|} D_n$, and 
$|\mathbb{K}|$ is the number of MCs in the set $\mathbb{K}$ i.e., number of MCs participated in FL process. 

After global model aggregation, each MC updates local weights using the global model for the next iteration (line 25). MC continues training using the updated local model weights and uploads them again to the associated miners for the next iteration. This process repeats in every global iteration (line 2-26). Finally, we obtained the optimal weights, $\textbf{W}^{T}$ after the FL process that gives the minimum value of loss function as given in Eq. \eqref{maxcost}.

\begin{table}[t]
	\begin{center}
		\caption{System parameters}
		\label{Table:Parameters}
		\begin{tabular}{l|c}
			\textbf{Parameters} & \textbf{Details}\\
			\hline
			$S$ & 5\\
			$N$ & 10-100\\
			Iterations & $T = 15, In = 10$\\
            Cycle rate of MC $f_n$ & 1-2.6 GHz\\
            Cycles per sample $\beta_n$ \cite{yang2020energy}& $[1,3]$ x $10^4$\\
            Bandwidth $Q$ & 20 MHz\\
            Number of allocated PRBs $V$ \cite{9371426} & 1-10\\
            $SINR_{n,s}$ & 13-20 dB\\
            Transmit powers $\mu_n$ \cite{yang2020energy}& 1-10 dB\\
            $\kappa$ \cite{yang2020energy}& $10^{-28}$\\
            Batch size & 32\\
            $R$ & 10 units\\
            $\varphi$ & 1 units\\
            Model weight size $H$ & 3.776 Kbits\\
			Weight parameters & $\rho$ = 0.5, $\eta$ = 0.5\\
			Threshold $\tau^{th}$ & 24 mins\\
			Initial learning rate $\alpha$ & 0.01\\
			
            \hline
			\end{tabular}
		\end{center}
\end{table}


\begin{theorem}\label{theoremcommu}
Communication complexity of Algorithm \ref{fdBlock} is $O(T(|\mathbb{K}|^{2}|w|)$.
\end{theorem}
\begin{proof} 
We consider the number of model weights sent by every MC for the local model training to analyze the communication complexity of Algorithm \ref{fdBlock}.
Let $|w|$ be the number of weights in a local model at each MC. 
In FL, the communication of weights during upload (line 13) and download of local model weights (line 23) from the blockchain is $|w|$ and $|\mathbb{K}||w|$, respectively \cite{singh2019detailed}.
Therefore, communication per MC in each global iteration is $|w| + |\mathbb{K}||w|$ i.e., $(|\mathbb{K}|+1)|w|$. Since $|\mathbb{K}|$ number of MCs participated in the FL process,  the communication required for MCs in each global iteration is $|\mathbb{K}|(|\mathbb{K}|+1)|w|$.
Similarly, miners in the blockchain broadcast $|\mathbb{K}|$ local models (line 21) and at least $(S-1)|w|$ communication is required to broadcast a local model. Thus, communication per miner in each global iteration is $|\mathbb{K}|(S-1)|w|$. 
Lines 2-26 iterates over the total number of global iterations i.e., $T$.
Therefore, overall communication required in $T$ global iteration is $T(|\mathbb{K}|(|\mathbb{K}|+1)|w| + |\mathbb{K}|(S-1)|w|)$ i.e., $T(|\mathbb{K}|^{2}|w| + |\mathbb{K}||w|S)$.
Generally, the number of MCs participated in the FL process is more than the number of miners i.e., $|\mathbb{K}| \geq S$. As a result, the communication complexity of the Algorithm \ref{fdBlock} is $O(T(|\mathbb{K}|^{2}|w|)$.
\end{proof}

\section{Performance Analysis} \label{result}

In this section, we present the performance of our proposed BPFISH framework through simulation analysis. The proposed framework is simulated using Python 3.9 and Tensorflow 2.0 on Windows 10 Home PC with Intel(R) Core(TM) i7-10750H @ 2.60 GHz processor and 16 GB memory.

We evaluate the performance of the proposed BPFISH framework on the Chest X-Ray Images (Pneumonia) dataset consisting of 5856 Chest X-Ray images \cite{PaulMooneydataset}. The dataset is divided into 5270 training samples and 586 testing samples.
Training data samples are partitioned into 5270/$N$ equal parts, with one part given to each MC. 
The images in the dataset have different resolutions, while our model needs a fixed dimension. 
Therefore, the images are down-sampled to a fixed resolution of 150 × 150 images. 
Our model uses a neural network that contains fourteen layers — the first ten layers are convolutional layers and the last four are fully connected layers. We use ReLU activation function in each layer except sigmoid activation function in output layer. 
The first convolutional layer uses 16 filters of size 3 x 3 with a stride of 1 pixel. The number of filters doubles every two convolution layers. Every two convolution layers are followed by batch normalization and max-pooling layer. The fully-connected layers have 512, 128, 68 and 1 neurons. The 'dropout' \cite{10.5555/2627435.2670313} rates of 0.7, 0.5 and 0.3 are used in the first three fully-connected layers to prevent overfitting.
For the optimization of the convolutional neural network, we set the initial learning rate to 0.01 and reduced it by a factor of 0.3 once the loss stopped decreasing for two iterations. \begin{figure}[t]
    \centering
    \includegraphics[width=0.8\linewidth]{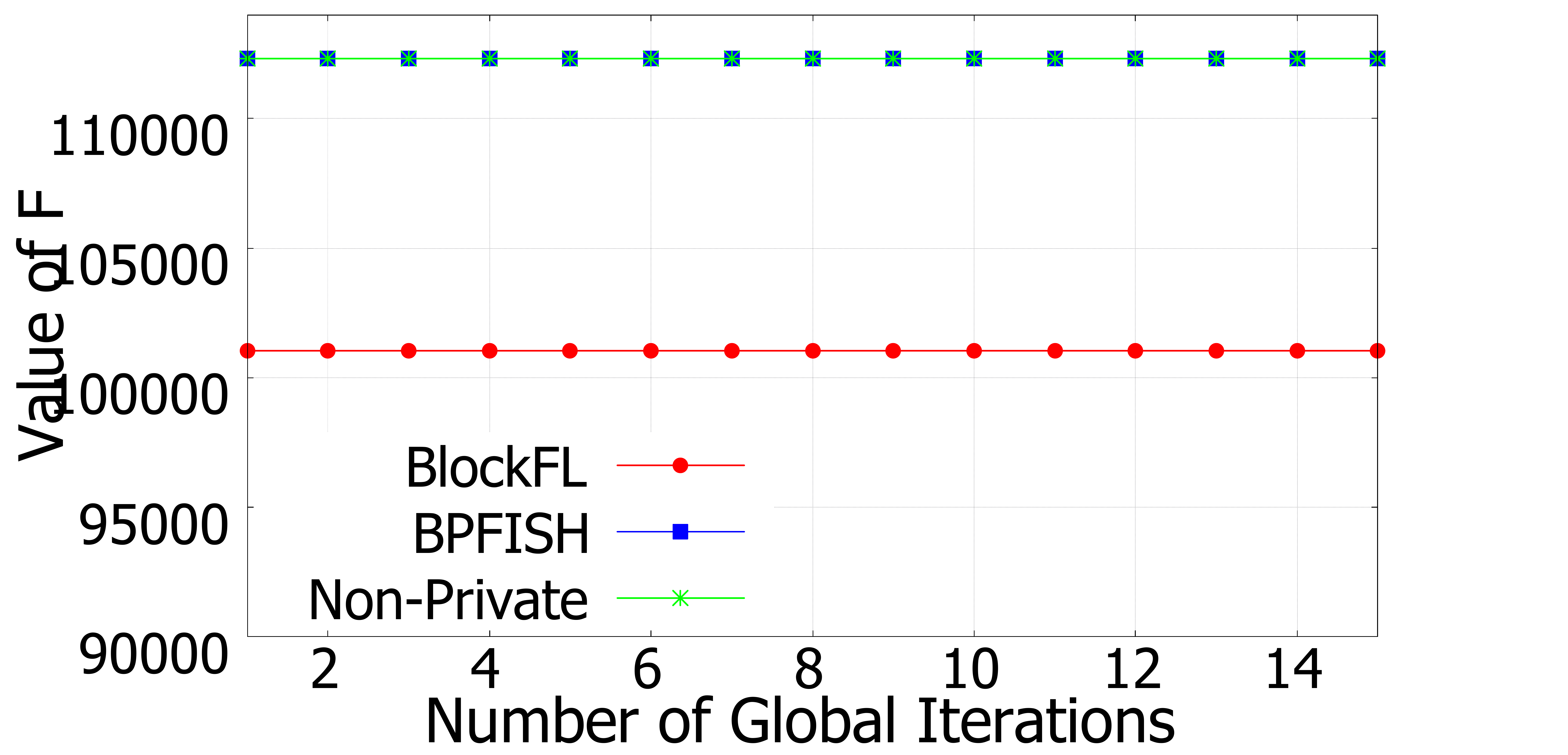}
    \caption{\label{utiweight}Analysis of the value of F.}
\end{figure}

In this experiment, cycle rate of each MC is set between 1 and 2.6 GHz, and cycles per sample $\beta_n$ are chosen uniformly from $[1, 3]$ x $10^{4}$ \cite{yang2020energy}.
Unless otherwise stated, we set coefficient $\kappa$ to $10^{-28}$ \cite{yang2020energy}, transmit powers $\mu_n$ from 1 to 10 dB \cite{yang2020energy}, model size $H$ = 3.776 Kbits, bandwidth $Q$ = 20 MHz and batch size $B$ = 32. Number of allocated PRBs $V$ are chosen from 1 to 10 and $SINR_{n,s}$ are chosen from 13 to 20 dB. 
The common parameters used in the experiment are shown in Table \ref{Table:Parameters}.

\subsection{Objective Function Analysis}

We plot the value of F as a function of global iteration, shown in Fig. \ref{utiweight}. 
We conduct the experiments by setting $N = 500$, $S = 50$, $T = 15$, $I_n = 10$, $\delta = 10^{-5}$, $\rho = 0.5$ and $\eta = 0.5$.
To get a fair comparison, we apply random association for BlockFL \cite{8733825} framework in which MCs select miners randomly. From the result, we can conclude that the proposed BPFISH is better compared to BlockFL and almost equal to non-private framework where non-private is the BPFISH without DP i.e., without adding noise to the gradient. The proposed BPFISH achieves 11.18\% better results on an average.
The reason for the better performance of BPFISH compared to BlockFL is that associations between miners and MCs in BPFISH are done more accurately using MMA algorithm in the proposed framework. The result of BPFISH is almost equal to that of non-private because both BPFISH and non-private framework used MMA algorithm and the difference in the value of loss function is small.

\subsection{Accuracy and Loss Comparison}

Fig. \ref{compp} compares the test accuracy and the value of loss function for 10 MCs with BlockFL and non-private framework using noise level $\sigma$ = 0.25 and gradient bound $A$ = 8.
As we can see from Fig. \ref{compacc}, the test accuracy increases as compared to BlockFL and the decrease in accuracy for private BPFISH as compared to the non-private framework is small and has little effect on the accuracy. 
In comparison to BlockFL, the proposed BPFISH obtains 10\% higher accuracy on average. 
As we can see from Fig. \ref{comploss}, the value of the loss function decrease as compared to BlockFL but is higher than the non-private. The increase in the value of loss for private BPFISH as compared to the non-private framework is small and has little effect on the value of the loss.

Table \ref{table:acc} shows the best test accuracy obtained when different noise levels are applied. 
As observed from the table, the proposed BPFISH achieved accuracy comparable to that of non-private collaborative learning framework when we added small noise. 
However, large noise affects test accuracy severely. 
Therefore, it is important to find an acceptable noise level to achieve strong DP without compromising accuracy.

\begin{figure}[!t]
 	\centering
 	\begin{subfigure}[b]{0.40\textwidth}
 		\centering
 		\includegraphics[width=\textwidth]{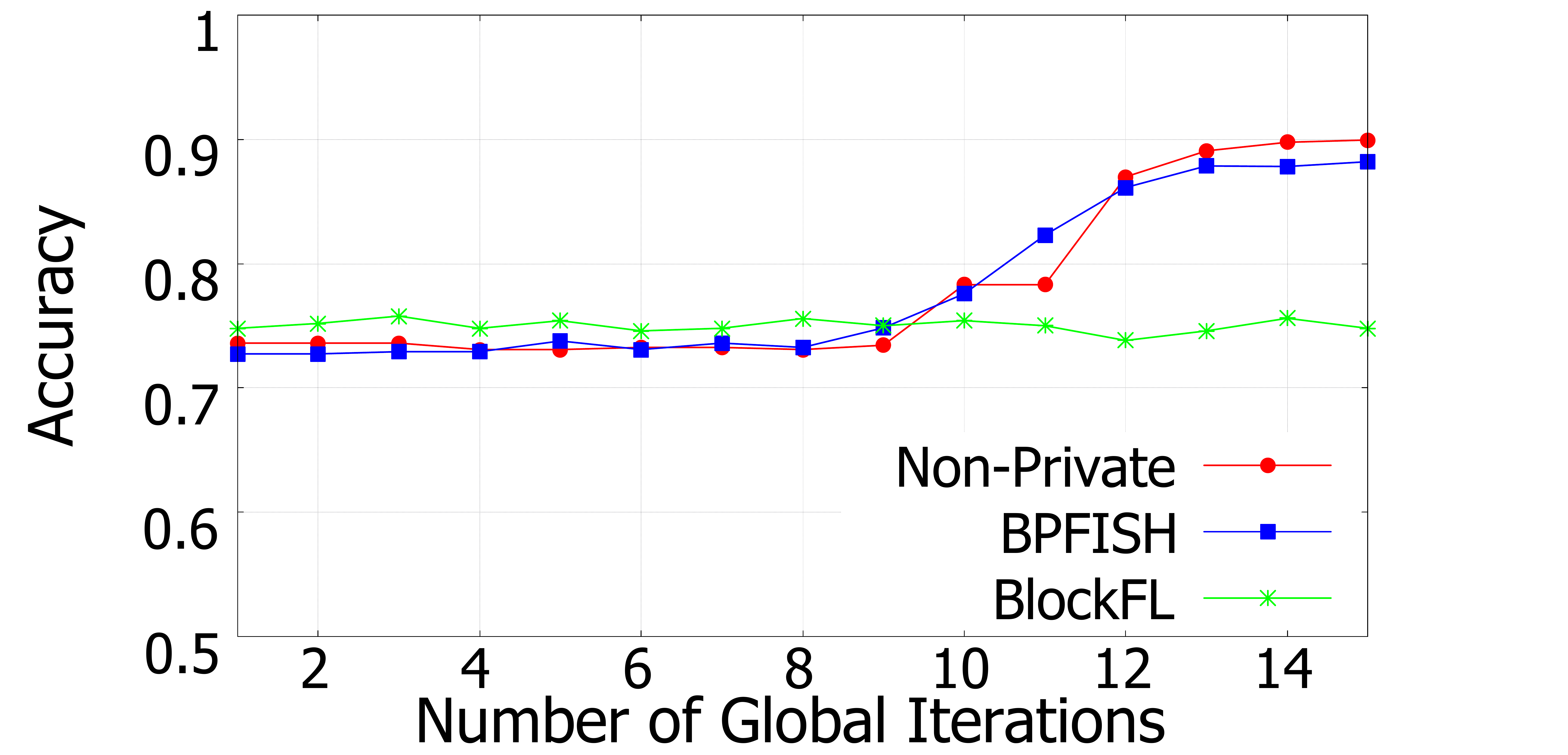}
 		\caption{Test accuracy comparison.}
 		\label{compacc}
 	\end{subfigure}
 	\begin{subfigure}[b]{0.40\textwidth}
 		\centering
 		\includegraphics[width=\textwidth]{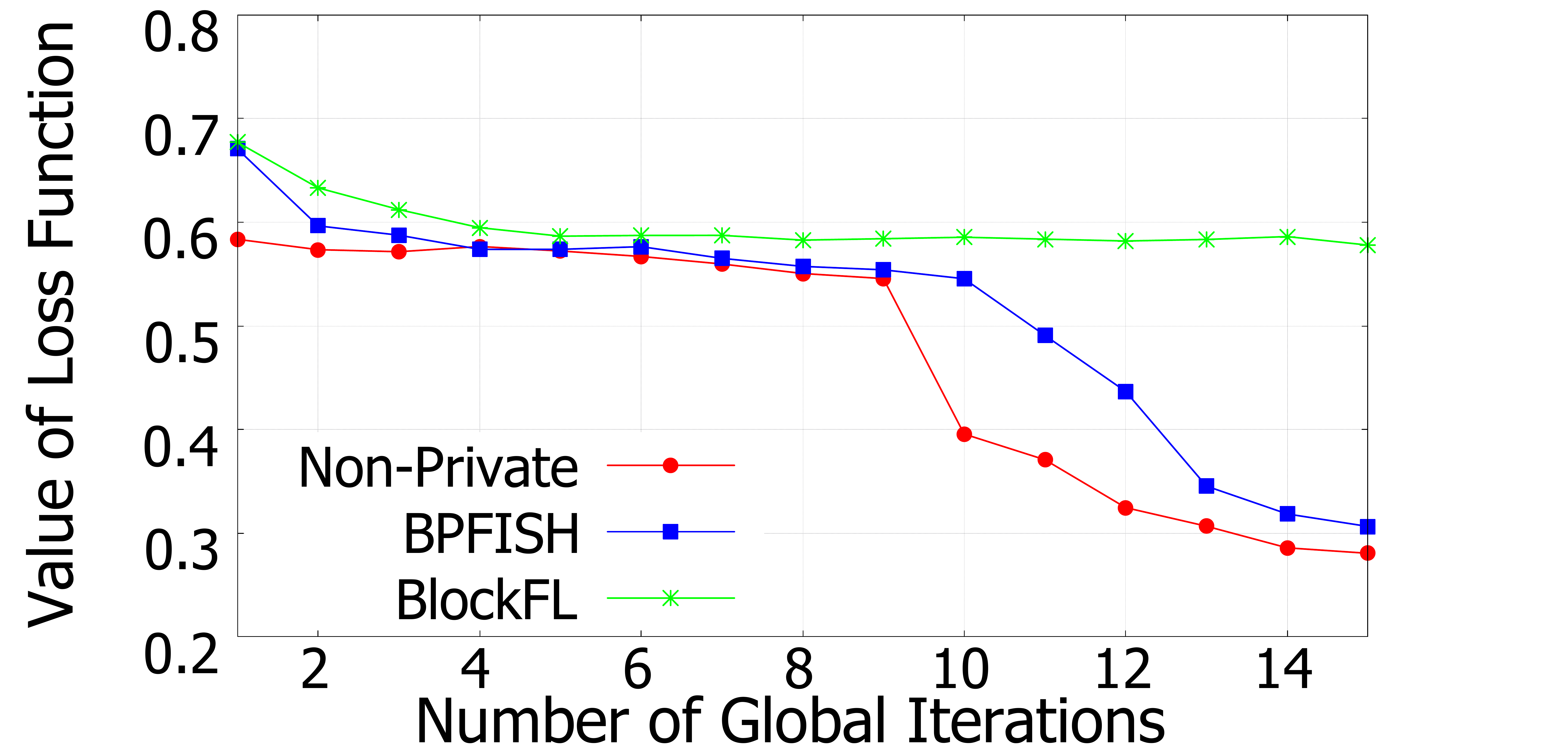}
 		\caption{Value of loss function comparison.}
 		\label{comploss}
 	\end{subfigure}
  	\caption{Test accuracy and loss analysis.}
  	\label{compp}
\end{figure}
\begin{figure}[t]
    \centering
    \includegraphics[width=0.83\linewidth]{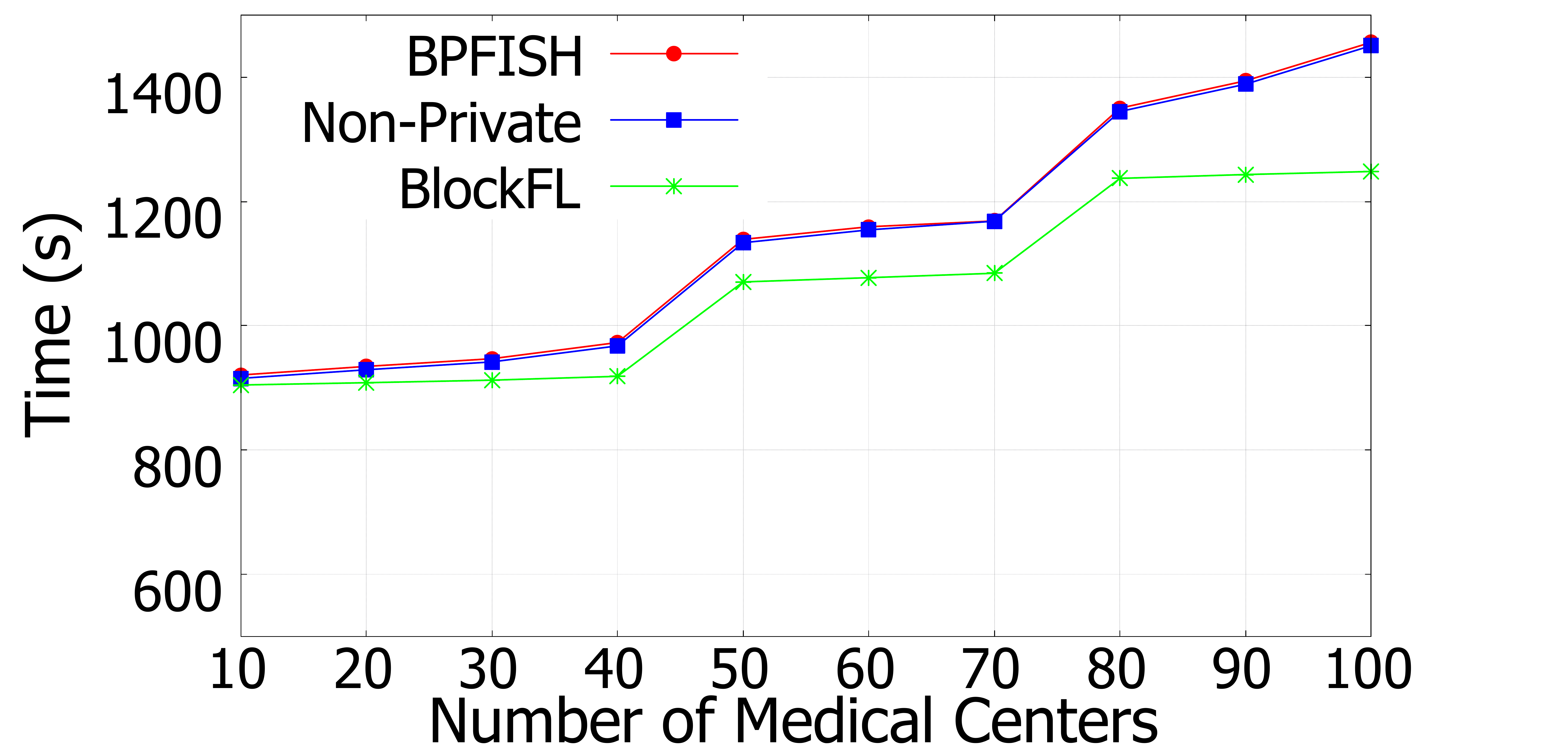}
    \caption{\label{timerand}Analysis of computation time.}
\end{figure}

\begin{figure*}[ht]
\centering
    \minipage{0.25\textwidth}
      \includegraphics[width=\textwidth]{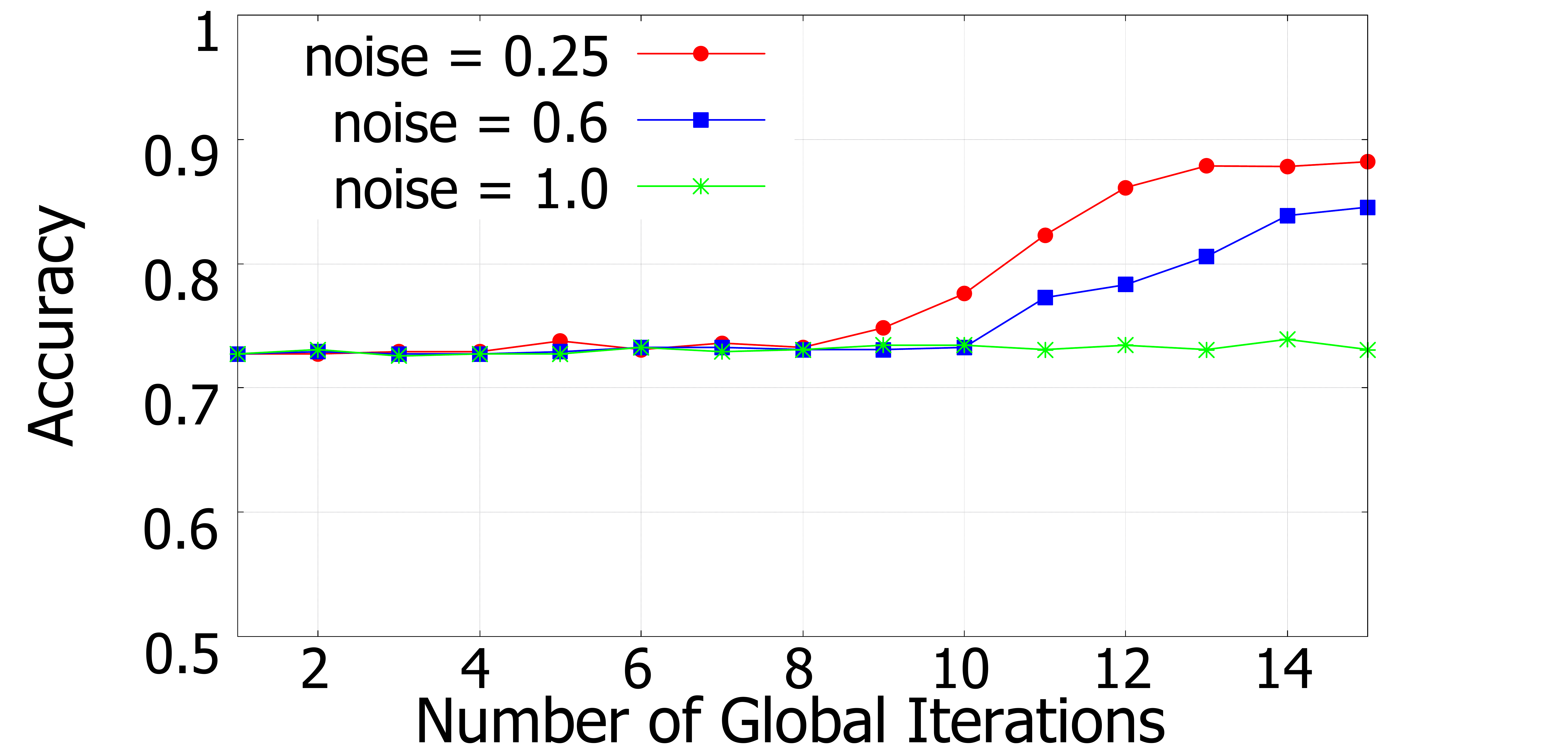}
      \subcaption{}
      \label{noiseacc}
    \endminipage\hfill
    \minipage{0.25\textwidth}
      \includegraphics[width=\textwidth]{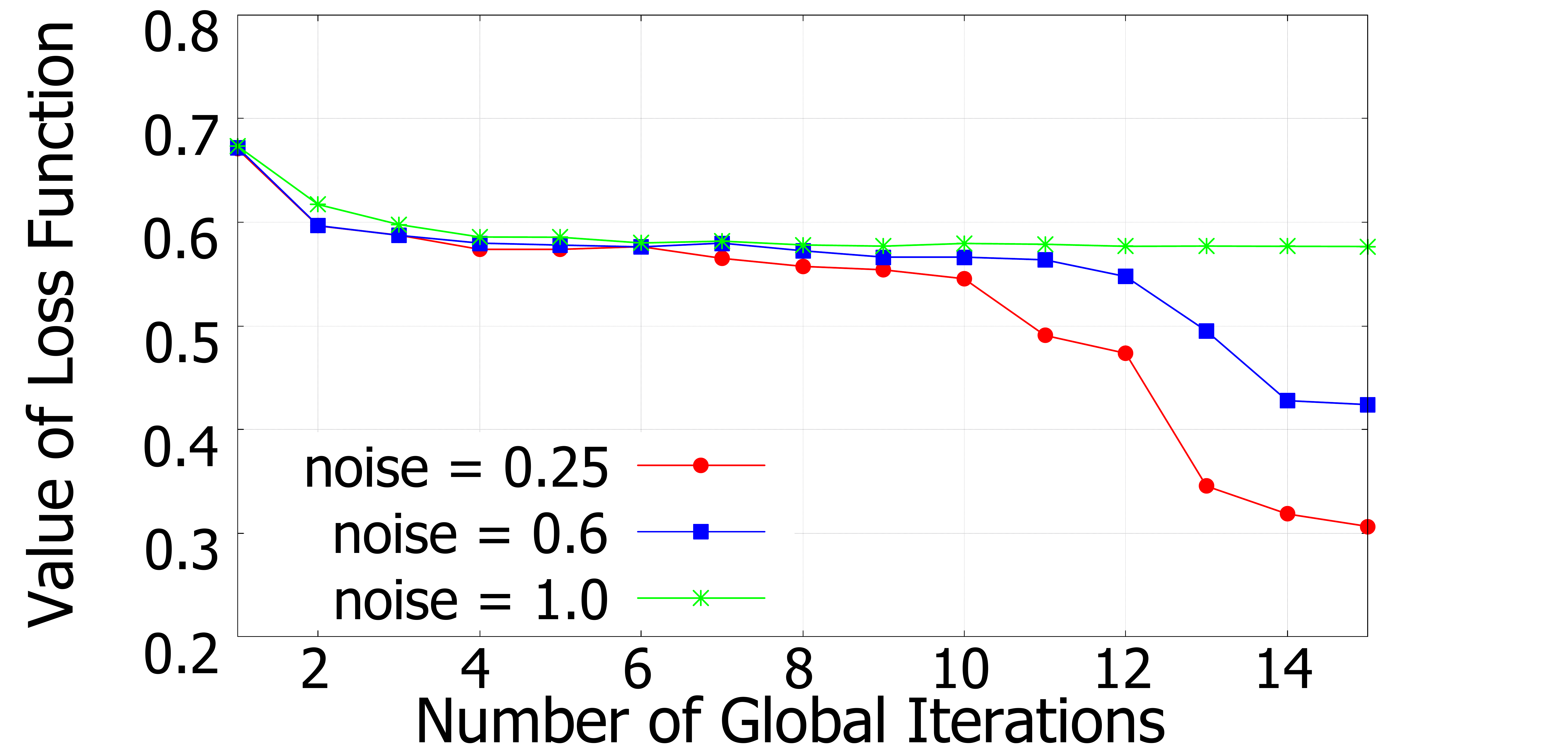}
      \subcaption{}
      \label{noiseloss}
    \endminipage\hfill
    \minipage{0.25\textwidth}%
        \includegraphics[width=\textwidth]{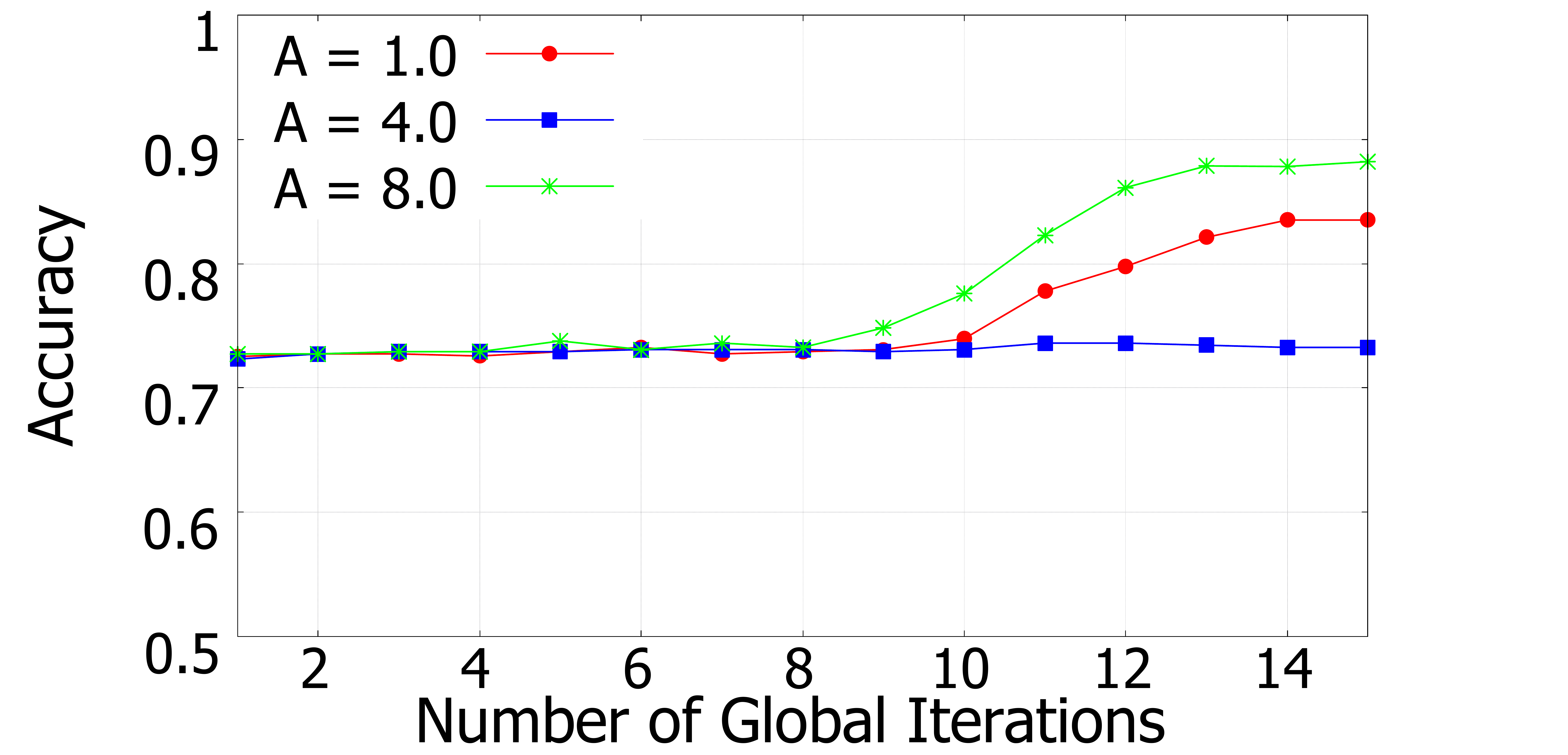}
      \subcaption{}
      \label{clipacc}
    \endminipage
    \minipage{0.25\textwidth}%
        \includegraphics[width=\textwidth]{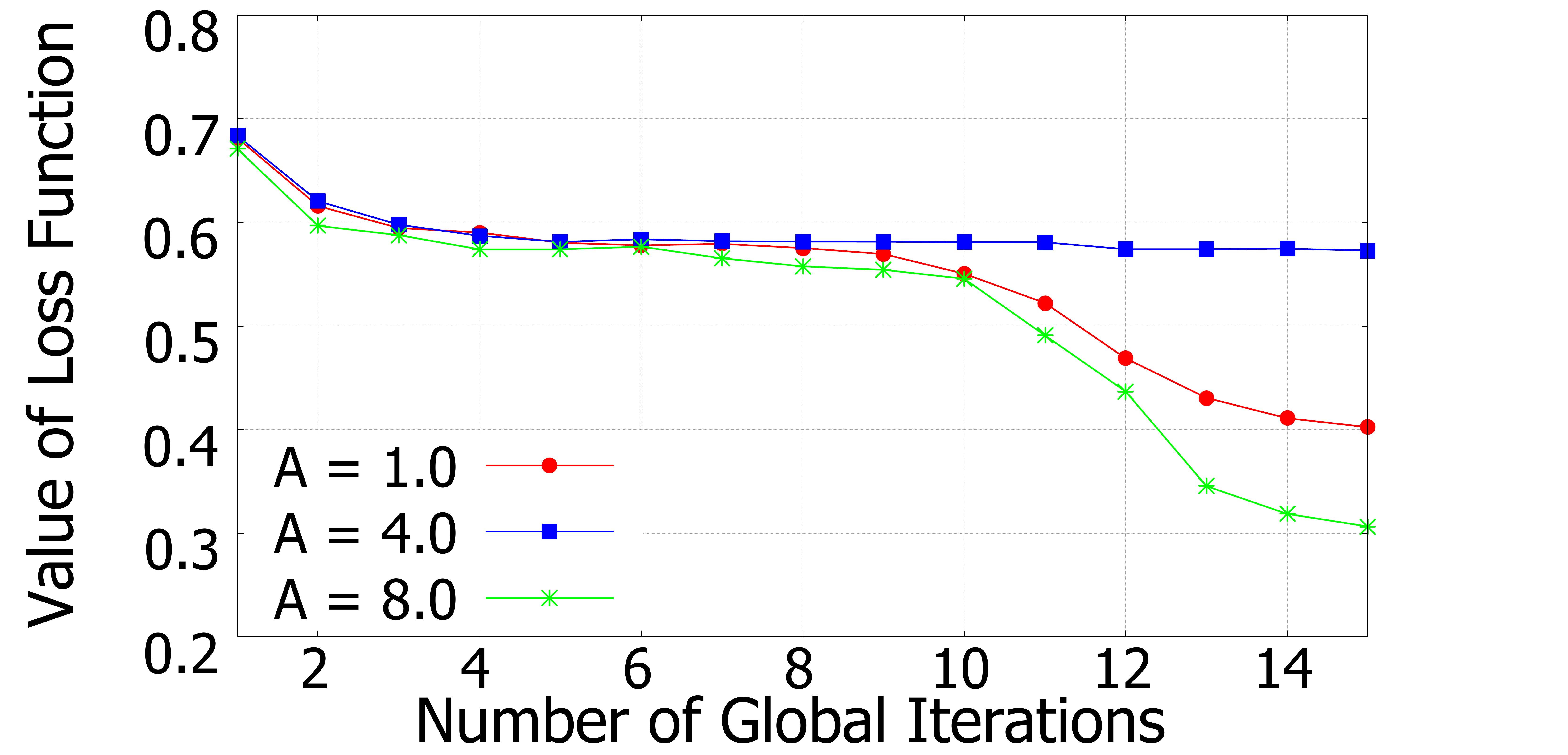}
      \subcaption{}
      \label{cliploss}
    \endminipage
  \caption{a) Effect of noise levels on test accuracy, b) Effect of noise levels on value of loss function, c) Effect of gradient bound on test accuracy, d) Effect of gradient bound on value of loss function.}
  \label{accloss}
\end{figure*}

\subsection{Computation Time Analysis}

Fig. \ref{timerand} shows computational time complexity analysis of the proposed BPFISH in comparison with BlockFL. For comparing computational time complexity, we set the number of MCs from 10 to 100 and the number of miners to 5.
As we can conclude from the figure, the time taken by the BPFISH is increasing with the increased number of MCs. The reason is that the number of local models added to the blockchain increases as the number of MCs increases.
The time taken of BPFISH is higher as compared to BlockFL because of the association formed by the MMA algorithm and the addition of noise in gradient calculation during local model training to ensure DP.
However, the time taken of BPFISH is almost equal to the time taken of non-private because the addition of noise takes less time.

\begin{table}[t]
    \begin{center}
        \caption{Accuracy for different privacy levels} \label{table:acc}
        \begin{tabular}{| c | c | c | c | c | c | c |}
        \hline
        \textbf{{$(\epsilon, \delta)$}} &
        \textbf{{Noise level}} &
        \textbf{{Accuracy}} \\
    
        \hline
            \textbf{{None}} & $\sigma = 0.00$ & $0.8993$ \\ 
            \hline
            {$(185, 10^{-5})$} & $\sigma = 0.25$ & $0.8819$ \\ 
            \hline
            {$(8, 10^{-5})$} & $\sigma = 0.60$ & $0.8454$  \\ 
            \hline
            {$(1.89, 10^{-5})$} & $\sigma = 1.00$ & $0.7309$  \\ 
            \hline
        \end{tabular}
    \end{center}
\end{table}

\subsection{Effects of Different Privacy Parameters}

In Fig. \ref{accloss}, we set $N = 10$, $S = 5$, $T = 15$, $I_n = 10$ and $\delta = 10^{-5}$.
We take various noise levels such as $\sigma$ = 0.25, $\sigma$ = 0.6 and $\sigma$ = 1.0 to illustrate the test accuracy and the value of loss function using gradient bound $A$ = 8 as shown in Figs. \ref{noiseacc} and \ref{noiseloss}.
We plot the test accuracy on various noise levels as a function of global iteration as shown in Fig \ref{noiseacc}. 
As seen from the figure, the accuracy increases as the noise level decreases and attained the highest accuracy when the noise level is set to $\sigma$ = 0.25.
Fig. \ref{noiseloss} compares the value of test set loss on various noise levels as a function of global iteration. We can see from the figure that the value of the loss function decreases as the noise level decreases and attained the best value when the noise level is set to $\sigma$  = 0.25.

In Figs. \ref{clipacc} and \ref{cliploss}, we set different gradient bounds such as $A$ = 1, $A$ = 4 and $A$ = 8 to illustrate the results of the test accuracy and the value of loss function using noise level $\sigma = 0.25$. Limiting the gradient bound destroys the true gradient value. Gradient bound destroys the true direction of gradient estimate if gradient bound is too small whereas a large gradient bound does not destroy true gradient. Therefore the clipped gradient becomes closer to true gradient estimates. 
As we can see from Fig. \ref{clipacc}, the test accuracy increases when the gradient bound increases from 1 to 8 and attained the best accuracy value when $A$ = 8.
Fig. \ref{cliploss} shows that the value of the test set loss function decreases when gradient bound increases from 1 to 8 and convergence performance of the proposed BPFISH obtained the best value when $A$ = 8.
We cannot compare BlockFL with BPFISH on various privacy parameters because BlockFL did not consider privacy parameters in their model.

\section{Conclusion} \label{conc}

In this paper, we have focused on decentralized privacy-preserving FL framework for learning effective models on healthcare data stored at different MCs.
We have proposed a joint optimization problem as maximization of utility and minimization of FL loss function altogether in smart healthcare domain. 
We introduced a stable association algorithm to maximize the utility of miners and MCs in polynomial time complexity. 
Moreover, we leveraged blockchain technology to enable tempered resistant and decentralized local model weights sharing. 
Through simulation analysis using Chest X-Ray Images (Pneumonia) dataset, we have
verified the effectiveness of the proposed BPFISH framework on various privacy parameters.
The simulation results demonstrated that BPFISH achieves high accuracy under a strong privacy protection level. 
Performance study shows that our proposed BPFISH framework outperforms state-of-the-art methods, achieving 11.18\% better results on an average.

In the future, we would like to consider other privacy-preserving techniques such as secure multi-party computation or homomorphic encryption. We would also like to work on energy efficient blockchain-based FL framework to reduce energy consumption during the FL process.
In addition, we will consider state-of-the-art neural network architectures and different data augmentation techniques to observe the effects of privacy parameters on accuracy.

\noindent \textbf{Acknowledgments:} This work is supported by the Science and Engineering Research Board (SERB), Government of India under Grant SRG/2020/000318.

\ifCLASSOPTIONcaptionsoff
  \newpage
\fi

\typeout{}
\bibliographystyle{IEEEtran}

\begin{thebibliography}{10}

\bibitem{li2017collaborative}
X.~Li, X.~Wang, K.~Li, Z.~Han, and V.~C. Leung, ``Collaborative multi-tier
  caching in heterogeneous networks: Modeling, analysis, and design,'' {\em
  IEEE Transactions on Wireless Communications}, 2017.

\bibitem{wang2015backhauling}
N.~Wang, E.~Hossain, and V.~K. Bhargava, ``Backhauling 5g small cells: A radio
  resource management perspective,'' {\em IEEE Wireless Communications},
  vol.~22, no.~5, pp.~41--49, 2015.

\bibitem{hasan2011green}
Z.~Hasan, H.~Boostanimehr, and V.~K. Bhargava, ``Green cellular networks: A
  survey, some research issues and challenges,'' {\em IEEE Communications
  surveys \& tutorials}, vol.~13, no.~4, pp.~524--540, 2011.

\bibitem{kim2010performance}
Y.~Kim, S.~Lee, and D.~Hong, ``Performance analysis of two-tier femtocell
  networks with outage constraints,'' {\em IEEE Transactions on Wireless
  Communications}, vol.~9, no.~9, pp.~2695--2700, 2010.

\bibitem{garcia2009autonomous}
L.~G. Garcia, K.~I. Pedersen, and P.~E. Mogensen, ``Autonomous component
  carrier selection: interference management in local area environments for
  lte-advanced,'' {\em IEEE Communications Magazine}, vol.~47, no.~9,
  pp.~110--116, 2009.

\bibitem{liang2012resource}
Y.-S. Liang, W.-H. Chung, G.-K. Ni, Y.~Chen, H.~Zhang, and S.-Y. Kuo,
  ``Resource allocation with interference avoidance in ofdma femtocell
  networks,'' {\em IEEE Transactions on Vehicular Technology}, vol.~61, no.~5,
  pp.~2243--2255, 2012.

\bibitem{pateromichelakis2013dynamic}
E.~Pateromichelakis, M.~Shariat, A.~Quddus, M.~Dianati, and R.~Tafazolli,
  ``Dynamic clustering framework for multi-cell scheduling in dense small cell
  networks,'' {\em IEEE Communications Letters}, vol.~17, no.~9,
  pp.~1802--1805, 2013.

\bibitem{hatoum2011fcra}
A.~Hatoum, N.~Aitsaadi, R.~Langar, R.~Boutaba, and G.~Pujolle, ``Fcra:
  Femtocell cluster-based resource allocation scheme for ofdma networks,'' in
  {\em 2011 IEEE International Conference on Communications (ICC)}, pp.~1--6,
  IEEE, 2011.

\bibitem{hatoum2014cluster}
A.~Hatoum, R.~Langar, N.~Aitsaadi, R.~Boutaba, and G.~Pujolle, ``Cluster-based
  resource management in ofdma femtocell networks with qos guarantees,'' {\em
  IEEE Transactions on Vehicular Technology}, vol.~63, no.~5, pp.~2378--2391,
  2014.

\bibitem{zhang2013coloring}
Q.~Zhang, X.~Zhu, L.~Wu, and K.~Sandrasegaran, ``A coloring-based resource
  allocation for ofdma femtocell networks,'' in {\em 2013 IEEE Wireless
  Communications and Networking Conference (WCNC)}, pp.~673--678, IEEE, 2013.

\bibitem{huang2016interference}
G.~Huang and J.~Li, ``Interference mitigation for femtocell networks via
  adaptive frequency reuse,'' {\em IEEE Transactions on Vehicular Technology},
  vol.~65, no.~4, pp.~2413--2423, 2016.

\bibitem{uygungelen2011graph}
S.~Uygungelen, G.~Auer, and Z.~Bharucha, ``Graph-based dynamic frequency reuse
  in femtocell networks,'' in {\em Vehicular Technology Conference (VTC
  Spring), 2011 IEEE 73rd}, pp.~1--6, IEEE, 2011.

\bibitem{chang2009multicell}
R.~Y. Chang, Z.~Tao, J.~Zhang, and C.-C.~J. Kuo, ``Multicell ofdma downlink
  resource allocation using a graphic framework,'' {\em IEEE Transactions on
  Vehicular Technology}, vol.~58, no.~7, pp.~3494--3507, 2009.

\bibitem{pateromichelakis2014graph}
E.~Pateromichelakis, M.~Shariat, A.~U. Quddus, and R.~Tafazolli, ``Graph-based
  multicell scheduling in ofdma-based small cell networks,'' {\em IEEE Access},
  vol.~2, pp.~897--908, 2014.

\bibitem{cao2005distributed}
L.~Cao and H.~Zheng, ``Distributed spectrum allocation via local bargaining.,''
  in {\em SECON}, pp.~475--486, 2005.

\bibitem{wang2011distributed}
Y.~Wang, K.~Zheng, X.~Shen, and W.~Wang, ``A distributed resource allocation
  scheme in femtocell networks,'' in {\em Vehicular Technology Conference (VTC
  Spring), 2011 IEEE 73rd}, pp.~1--5, IEEE, 2011.

\bibitem{anand2012maximum}
S.~Anand, S.~Sengupta, and R.~Chandramouli, ``Maximum spectrum packing: a
  distributed opportunistic channel acquisition mechanism in dynamic spectrum
  access networks,'' {\em IET communications}, vol.~6, no.~8, pp.~872--882,
  2012.

\bibitem{sundaresan2009efficient}
K.~Sundaresan and S.~Rangarajan, ``Efficient resource management in ofdma femto
  cells,'' in {\em Proceedings of the tenth ACM international symposium on
  Mobile ad hoc networking and computing}, pp.~33--42, ACM, 2009.

\bibitem{lu2013achieving}
Z.~Lu, T.~Bansal, and P.~Sinha, ``Achieving user-level fairness in open-access
  femtocell-based architecture,'' {\em IEEE transactions on mobile computing},
  vol.~12, no.~10, pp.~1943--1954, 2013.

\bibitem{yang2017spectrum}
Y.~Yang, B.~Bai, and W.~Chen, ``Spectrum reuse ratio in 5g cellular networks: A
  matrix graph approach,'' {\em IEEE Transactions on Mobile Computing}, 2017.

\bibitem{wubben2014benefits}
D.~Wubben, P.~Rost, J.~S. Bartelt, M.~Lalam, V.~Savin, M.~Gorgoglione,
  A.~Dekorsy, and G.~Fettweis, ``Benefits and impact of cloud computing on 5g
  signal processing: Flexible centralization through cloud-ran,'' {\em IEEE
  signal processing magazine}, vol.~31, no.~6, pp.~35--44, 2014.

\bibitem{checko2015cloud}
A.~Checko, H.~L. Christiansen, Y.~Yan, L.~Scolari, G.~Kardaras, M.~S. Berger,
  and L.~Dittmann, ``Cloud ran for mobile networks—a technology overview,''
  {\em IEEE Communications surveys \& tutorials}, vol.~17, no.~1, pp.~405--426,
  2015.

\bibitem{A}
 3GPP standardization, ``Evolved Universal Terrestrial Radio Access (E-UTRA)
  and Evolved Universal Terrestrial Radio Access Network (E-UTRAN); Overall
  Description; Stage 2,'' TS 36.300 V10.5.0, 2011.

\bibitem{garey2002computers}
M.~R. Garey and D.~S. Johnson, {\em Computers and intractability}, vol.~29.
\newblock wh freeman New York, 2002.

\bibitem{ahmed2008online}
N.~Ahmed, U.~Ismail, S.~Keshav, and K.~Papagiannaki, ``Online estimation of rf
  interference,'' in {\em Proceedings of the 2008 ACM CoNEXT Conference}, p.~4,
  ACM, 2008.

\bibitem{kuhn2009local}
F.~Kuhn, ``Local multicoloring algorithms: Computing a nearly-optimal tdma
  schedule in constant time,'' {\em arXiv preprint arXiv:0902.1868}, 2009.

\bibitem{kuhn2006complexity}
F.~Kuhn and R.~Wattenhofer, ``On the complexity of distributed graph
  coloring,'' in {\em Proceedings of the twenty-fifth annual ACM symposium on
  Principles of distributed computing}, pp.~7--15, ACM, 2006.

\bibitem{A2}
 Evolved Universal Terrestrial Radio Access (E-UTRA); Physical layer procedures
  (version 10.4.0), 3GPP Std. TS 36.213, Dec. 2011.

\bibitem{baker2011uplink}
M.~Baker, ``Uplink transmission procedures,'' {\em LTE-The UMTS Long Term
  Evolution: From Theory to Practice, Second Edition}, pp.~407--420, 2011.

\bibitem{jain1984quantitative}
R.~Jain, D.-M. Chiu, and W.~R. Hawe, {\em A quantitative measure of fairness
  and discrimination for resource allocation in shared computer system},
  vol.~38.
\newblock Eastern Research Laboratory, Digital Equipment Corporation Hudson,
  MA, 1984.

\end{thebibliography}


\begin{thebibliography}{10}
\providecommand{\url}[1]{#1}
\csname url@samestyle\endcsname
\providecommand{\newblock}{\relax}
\providecommand{\bibinfo}[2]{#2}
\providecommand{\BIBentrySTDinterwordspacing}{\spaceskip=0pt\relax}
\providecommand{\BIBentryALTinterwordstretchfactor}{4}
\providecommand{\BIBentryALTinterwordspacing}{\spaceskip=\fontdimen2\font plus
\BIBentryALTinterwordstretchfactor\fontdimen3\font minus
  \fontdimen4\font\relax}
\providecommand{\BIBforeignlanguage}[2]{{%
\expandafter\ifx\csname l@#1\endcsname\relax
\typeout{** WARNING: IEEEtran.bst: No hyphenation pattern has been}%
\typeout{** loaded for the language `#1'. Using the pattern for}%
\typeout{** the default language instead.}%
\else
\language=\csname l@#1\endcsname
\fi
#2}}
\providecommand{\BIBdecl}{\relax}
\BIBdecl

\bibitem{singh2020internet}
R.~P. Singh, M.~Javaid, A.~Haleem, and R.~Suman, ``Internet of things ({IoT})
  applications to fight against covid-19 pandemic,'' \emph{Diabetes \&
  Metabolic Syndrome: Clinical Research \& Reviews}, vol.~14, no.~4, pp.
  521--524, 2020.

\bibitem{jnr2020use}
B.~A. Jnr, ``Use of telemedicine and virtual care for remote treatment in
  response to covid-19 pandemic,'' \emph{Journal of Medical Systems}, vol.~44,
  no.~7, pp. 1--9, 2020.

\bibitem{10.1145/3298981}
\BIBentryALTinterwordspacing
Q.~Yang, Y.~Liu, T.~Chen, and Y.~Tong, ``Federated machine learning: Concept
  and applications,'' \emph{ACM Trans. Intell. Syst. Technol.}, vol.~10, no.~2,
  jan 2019. [Online]. Available: \url{https://doi.org/10.1145/3298981}
\BIBentrySTDinterwordspacing

\bibitem{mcmahan2017communication}
B.~McMahan, E.~Moore, D.~Ramage, S.~Hampson, and B.~A. y~Arcas,
  ``Communication-efficient learning of deep networks from decentralized
  data,'' in \emph{Artificial intelligence and statistics}.\hskip 1em plus
  0.5em minus 0.4em\relax PMLR, 2017, pp. 1273--1282.

\bibitem{KimberlyPowellBlog}
\BIBentryALTinterwordspacing
K.~Powell. {NVIDIA} clara federated learning to deliver {AI} to hospitals while
  protecting patient data: Intelligent edge computing platform streamlines deep
  learning for radiology. [Online]. Available:
  \url{https://blogs.nvidia.com/blog/2019/12/01/clara-federated-learning/}
\BIBentrySTDinterwordspacing

\bibitem{8931716}
A.~I. Newaz, A.~K. Sikder, M.~A. Rahman, and A.~S. Uluagac, ``Healthguard: A
  machine learning-based security framework for smart healthcare systems,'' in
  \emph{2019 Sixth International Conference on Social Networks Analysis,
  Management and Security (SNAMS)}, 2019, pp. 389--396.

\bibitem{li2020federated}
T.~Li, A.~K. Sahu, A.~Talwalkar, and V.~Smith, ``Federated learning:
  Challenges, methods, and future directions,'' \emph{IEEE Signal Processing
  Magazine}, vol.~37, no.~3, pp. 50--60, 2020.

\bibitem{8733825}
H.~Kim, J.~Park, M.~Bennis, and S.-L. Kim, ``Blockchained on-device federated
  learning,'' \emph{IEEE Communications Letters}, vol.~24, no.~6, pp.
  1279--1283, 2020.

\bibitem{10.1007/11681878_14}
C.~Dwork, F.~McSherry, K.~Nissim, and A.~Smith, ``Calibrating noise to
  sensitivity in private data analysis,'' in \emph{Theory of Cryptography},
  S.~Halevi and T.~Rabin, Eds.\hskip 1em plus 0.5em minus 0.4em\relax Berlin,
  Heidelberg: Springer Berlin Heidelberg, 2006, pp. 265--284.

\bibitem{Abadi_2016}
\BIBentryALTinterwordspacing
M.~Abadi, A.~Chu, I.~Goodfellow, H.~B. McMahan, I.~Mironov, K.~Talwar, and
  L.~Zhang, ``Deep learning with differential privacy,'' \emph{Proceedings of
  the 2016 ACM SIGSAC Conference on Computer and Communications Security}, Oct
  2016. [Online]. Available: \url{http://dx.doi.org/10.1145/2976749.2978318}
\BIBentrySTDinterwordspacing

\bibitem{wang2019collecting}
N.~Wang, X.~Xiao, Y.~Yang, J.~Zhao, S.~C. Hui, H.~Shin, J.~Shin, and G.~Yu,
  ``Collecting and analyzing multidimensional data with local differential
  privacy,'' in \emph{2019 IEEE 35th International Conference on Data
  Engineering (ICDE)}.\hskip 1em plus 0.5em minus 0.4em\relax IEEE, 2019, pp.
  638--649.

\bibitem{tuli2020healthfog}
S.~Tuli, N.~Basumatary, S.~S. Gill, M.~Kahani, R.~C. Arya, G.~S. Wander, and
  R.~Buyya, ``Healthfog: An ensemble deep learning based smart healthcare
  system for automatic diagnosis of heart diseases in integrated {IoT} and fog
  computing environments,'' \emph{Future Generation Computer Systems}, vol.
  104, pp. 187--200, 2020.

\bibitem{chen2020fedhealth}
Y.~Chen, X.~Qin, J.~Wang, C.~Yu, and W.~Gao, ``Fedhealth: A federated transfer
  learning framework for wearable healthcare,'' \emph{IEEE Intelligent
  Systems}, vol.~35, no.~4, pp. 83--93, 2020.

\bibitem{silva2019federated}
S.~Silva, B.~A. Gutman, E.~Romero, P.~M. Thompson, A.~Altmann, and M.~Lorenzi,
  ``Federated learning in distributed medical databases: Meta-analysis of
  large-scale subcortical brain data,'' in \emph{2019 IEEE 16th international
  symposium on biomedical imaging (ISBI 2019)}.\hskip 1em plus 0.5em minus
  0.4em\relax IEEE, 2019, pp. 270--274.

\bibitem{sadilek2021privacy}
A.~Sadilek, L.~Liu, D.~Nguyen, M.~Kamruzzaman, S.~Serghiou, B.~Rader,
  A.~Ingerman, S.~Mellem, P.~Kairouz, E.~O. Nsoesie \emph{et~al.},
  ``Privacy-first health research with federated learning,'' \emph{NPJ digital
  medicine}, vol.~4, no.~1, pp. 1--8, 2021.

\bibitem{9492000}
J.~Li, Y.~Meng, L.~Ma, S.~Du, H.~Zhu, Q.~Pei, and S.~Shen, ``A federated
  learning based privacy-preserving smart healthcare system,'' \emph{IEEE
  Transactions on Industrial Informatics}, pp. 1--1, 2021.

\bibitem{LI2020101765}
\BIBentryALTinterwordspacing
X.~Li, Y.~Gu, N.~Dvornek, L.~H. Staib, P.~Ventola, and J.~S. Duncan,
  ``Multi-site fmri analysis using privacy-preserving federated learning and
  domain adaptation: Abide results,'' \emph{Medical Image Analysis}, vol.~65,
  p. 101765, 2020. [Online]. Available:
  \url{https://www.sciencedirect.com/science/article/pii/S1361841520301298}
\BIBentrySTDinterwordspacing

\bibitem{9069945}
K.~Wei, J.~Li, M.~Ding, C.~Ma, H.~H. Yang, F.~Farokhi, S.~Jin, T.~Q.~S. Quek,
  and H.~V. Poor, ``Federated learning with differential privacy: Algorithms
  and performance analysis,'' \emph{IEEE Transactions on Information Forensics
  and Security}, vol.~15, pp. 3454--3469, 2020.

\bibitem{novo2018blockchain}
O.~Novo, ``Blockchain meets {IoT}: An architecture for scalable access
  management in {IoT},'' \emph{IEEE Internet of Things Journal}, vol.~5, no.~2,
  pp. 1184--1195, 2018.

\bibitem{zhang2018blockchain}
X.~Zhang and S.~Poslad, ``Blockchain support for flexible queries with granular
  access control to electronic medical records (emr),'' in \emph{2018 IEEE
  International conference on communications (ICC)}.\hskip 1em plus 0.5em minus
  0.4em\relax IEEE, 2018, pp. 1--6.

\bibitem{dagher2018ancile}
G.~G. Dagher, J.~Mohler, M.~Milojkovic, and P.~B. Marella, ``Ancile:
  Privacy-preserving framework for access control and interoperability of
  electronic health records using blockchain technology,'' \emph{Sustainable
  cities and society}, vol.~39, pp. 283--297, 2018.

\bibitem{9261429}
C.~Li, M.~Dong, J.~Li, G.~Xu, X.~Chen, and K.~Ota, ``Healthchain: Secure emrs
  management and trading in distributed healthcare service system,'' \emph{IEEE
  Internet of Things Journal}, vol.~8, no.~9, pp. 7192--7202, 2021.

\bibitem{9112693}
S.~Seng, C.~Luo, X.~Li, H.~Zhang, and H.~Ji, ``User matching on blockchain for
  computation offloading in ultra-dense wireless networks,'' \emph{IEEE
  Transactions on Network Science and Engineering}, vol.~8, no.~2, pp.
  1167--1177, 2021.

\bibitem{lim2021towards}
W.~Y.~B. Lim, J.~Huang, Z.~Xiong, J.~Kang, D.~Niyato, X.-S. Hua, C.~Leung, and
  C.~Miao, ``Towards federated learning in uav-enabled internet of vehicles: A
  multi-dimensional contract-matching approach,'' \emph{IEEE Transactions on
  Intelligent Transportation Systems}, 2021.

\bibitem{KADADHA2021103155}
\BIBentryALTinterwordspacing
M.~Kadadha, H.~Otrok, S.~Singh, R.~Mizouni, and A.~Ouali, ``Two-sided
  preferences task matching mechanisms for blockchain-based crowdsourcing,''
  \emph{Journal of Network and Computer Applications}, vol. 191, p. 103155,
  2021. [Online]. Available:
  \url{https://www.sciencedirect.com/science/article/pii/S1084804521001697}
\BIBentrySTDinterwordspacing

\bibitem{8726129}
J.~An, H.~Yang, X.~Gui, W.~Zhang, R.~Gui, and J.~Kang, ``{TCNS}: Node selection
  with privacy protection in crowdsensing based on twice consensuses of
  blockchain,'' \emph{IEEE Transactions on Network and Service Management},
  vol.~16, no.~3, pp. 1255--1267, 2019.

\bibitem{9176324}
R.~Chen, M.~Chen, and H.~Yang, ``Dynamic physician-patient matching in the
  healthcare system,'' in \emph{2020 42nd Annual International Conference of
  the IEEE Engineering in Medicine Biology Society (EMBC)}, 2020, pp.
  5868--5871.

\bibitem{lu2019blockchain}
Y.~Lu, X.~Huang, Y.~Dai, S.~Maharjan, and Y.~Zhang, ``Blockchain and federated
  learning for privacy-preserved data sharing in industrial {IoT},'' \emph{IEEE
  Transactions on Industrial Informatics}, vol.~16, no.~6, pp. 4177--4186,
  2019.

\bibitem{8874972}
K.~Gai, Y.~Wu, L.~Zhu, Z.~Zhang, and M.~Qiu, ``Differential privacy-based
  blockchain for industrial internet-of-things,'' \emph{IEEE Transactions on
  Industrial Informatics}, vol.~16, no.~6, pp. 4156--4165, 2020.

\bibitem{jia2021blockchain}
B.~Jia, X.~Zhang, J.~Liu, Y.~Zhang, K.~Huang, and Y.~Liang,
  ``Blockchain-enabled federated learning data protection aggregation scheme
  with differential privacy and homomorphic encryption in {IIoT},'' \emph{IEEE
  Transactions on Industrial Informatics}, 2021.

\bibitem{ramanan2020baffle}
P.~Ramanan and K.~Nakayama, ``{BAFFLE}: Blockchain based aggregator free
  federated learning,'' in \emph{2020 IEEE International Conference on
  Blockchain (Blockchain)}.\hskip 1em plus 0.5em minus 0.4em\relax IEEE, 2020,
  pp. 72--81.

\bibitem{article}
S.~Nakamoto, ``Bitcoin: A peer-to-peer electronic cash system,''
  \emph{Cryptography Mailing list at https://metzdowd.com}, 03 2009.

\bibitem{9264742}
Z.~Yang, M.~Chen, W.~Saad, C.~S. Hong, and M.~Shikh-Bahaei, ``Energy efficient
  federated learning over wireless communication networks,'' \emph{IEEE
  Transactions on Wireless Communications}, vol.~20, no.~3, pp. 1935--1949,
  2021.

\bibitem{9371426}
A.~Pratap and S.~K. Das, ``Stable matching based resource allocation for
  service provider's revenue maximization in {5G} networks,'' \emph{IEEE
  Transactions on Mobile Computing}, pp. 1--1, 2021.

\bibitem{kang2019incentive}
J.~Kang, Z.~Xiong, D.~Niyato, S.~Xie, and J.~Zhang, ``Incentive mechanism for
  reliable federated learning: A joint optimization approach to combining
  reputation and contract theory,'' \emph{IEEE Internet of Things Journal},
  vol.~6, no.~6, pp. 10\,700--10\,714, 2019.

\bibitem{10.1145/3410566.3410598}
\BIBentryALTinterwordspacing
C.~K. Leung, D.~L.~X. Fung, S.~B. Mushtaq, O.~T. Leduchowski, R.~L. Bouchard,
  H.~Jin, A.~Cuzzocrea, and C.~Y. Zhang, ``Data science for healthcare
  predictive analytics,'' in \emph{Proceedings of the 24th Symposium on
  International Database Engineering 'I\&' Applications}, ser. IDEAS '20.\hskip
  1em plus 0.5em minus 0.4em\relax New York, NY, USA: Association for Computing
  Machinery, 2020. [Online]. Available:
  \url{https://doi.org/10.1145/3410566.3410598}
\BIBentrySTDinterwordspacing

\bibitem{reisizadeh2020fedpaq}
A.~Reisizadeh, A.~Mokhtari, H.~Hassani, A.~Jadbabaie, and R.~Pedarsani,
  ``Fedpaq: A communication-efficient federated learning method with periodic
  averaging and quantization,'' in \emph{International Conference on Artificial
  Intelligence and Statistics}.\hskip 1em plus 0.5em minus 0.4em\relax PMLR,
  2020, pp. 2021--2031.

\bibitem{gale1962college}
D.~Gale and L.~S. Shapley, ``College admissions and the stability of
  marriage,'' \emph{The American Mathematical Monthly}, vol.~69, no.~1, pp.
  9--15, 1962.

\bibitem{yang2020energy}
Z.~Yang, M.~Chen, W.~Saad, C.~S. Hong, and M.~Shikh-Bahaei, ``Energy efficient
  federated learning over wireless communication networks,'' \emph{IEEE
  Transactions on Wireless Communications}, vol.~20, no.~3, pp. 1935--1949,
  2020.

\bibitem{singh2019detailed}
A.~Singh, P.~Vepakomma, O.~Gupta, and R.~Raskar, ``Detailed comparison of
  communication efficiency of split learning and federated learning,''
  \emph{arXiv preprint arXiv:1909.09145}, 2019.

\bibitem{PaulMooneydataset}
\BIBentryALTinterwordspacing
P.~Mooney. Chest x-ray images (pneumonia), version 2. [Online]. Available:
  \url{https://www.kaggle.com/paultimothymooney/chest-xray-pneumonia}
\BIBentrySTDinterwordspacing

\bibitem{10.5555/2627435.2670313}
N.~Srivastava, G.~Hinton, A.~Krizhevsky, I.~Sutskever, and R.~Salakhutdinov,
  ``Dropout: A simple way to prevent neural networks from overfitting,''
  \emph{J. Mach. Learn. Res.}, vol.~15, no.~1, p. 1929–1958, jan 2014.

\end{thebibliography}


\begin{IEEEbiography}[{\includegraphics[width=1.1in,height=1.2in,clip,keepaspectratio]{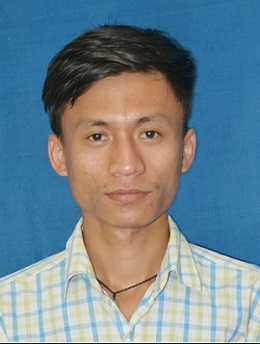}}]{Moirangthem Biken Singh} completed the B.Tech degree in Computer Science and Engineering from the National Institute of Technology Manipur, India, in 2018 and the M.Tech degree from National Institute of Technology Kurukshetra, India, in 2021. He is currently pursuing Ph.D. degree in Computer Science and Engineering, Indian Institute of Technology (BHU) Varanasi, India. His current research interest include AI, machine learning and FL in Smart Healthcare.

\end{IEEEbiography}

\begin{IEEEbiography}[{\includegraphics[width=1.1in,height=1.2in,clip,keepaspectratio]{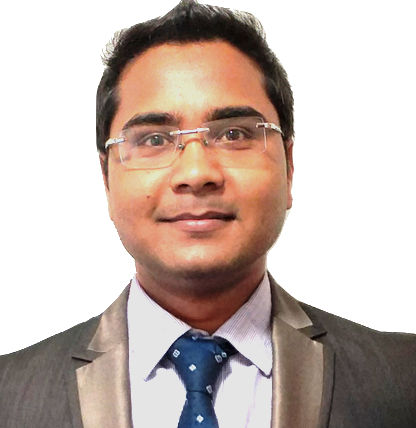}}]{Ajay Pratap} is an Assistant Professor with the Department of Computer Science and Engineering, Indian Institute of Technology (BHU) Varanasi, India. Before joining IIT (BHU), he was associated with the Department of Computer Science and Engineering, National Institute of Technology Karnataka (NITK) Surathkal, India, as an Assistant Professor from December 2019 to May 2020. He worked as a Postdoctoral Researcher in the Department of Computer Science at Missouri University of Science and Technology, USA, from August 2018 to December 2019. He completed his Ph.D. degree in Computer Science and Engineering from the Indian Institute of Technology Patna, India, in July 2018. His research interests include Cyber-Physical Systems, IoT-enabled Smart Environments, Mobile Computing and Networking, Statistical Learning, Algorithm Design for Next-generation Advanced Wireless Networks, Applied Graph Theory, and Game Theory. His current work is related to HetNet, Small Cells, Fog Computing, IoT, and D2D communication underlaying cellular 5G and beyond. His papers appeared in several international journals and conferences including IEEE Transactions on Mobile Computing, IEEE Transactions on Parallel and Distributed Systems, and IEEE LCN, etc. He has received several awards including the Best Paper Candidate Award and NSF travel grant for IEEE Smartcom'19 in the USA. 
\end{IEEEbiography}

\end{document}